\documentclass[12pt]{article}
\usepackage{bm,amsmath,amsthm,amssymb,graphicx,paralist,setspace}
\usepackage{color}
\usepackage[margin=1in]{geometry}
\usepackage{natbib}
\usepackage[colorlinks]{hyperref}
\usepackage{multicol}
\usepackage{enumitem}
\usepackage{setspace}
\usepackage{microtype}
\usepackage{booktabs}
\usepackage{multibib}
\newcites{supp}{Supplementary references}

\let\originalleft\left
\let\originalright\right
\renewcommand{\left}{\mathopen{}\mathclose\bgroup\originalleft}
\renewcommand{\right}{\aftergroup\egroup\originalright}

\def\bas#1\eas{\begin{align*}#1\end{align*}}
\def\basn#1\easn{\begin{align}#1\end{align}}

\def\highlightchangessince{99}
\definecolor{changecolor}{RGB}{85,85,255}
\def\bct#1#2\ect#3{\ifnum#1>\highlightchangessince\textcolor{changecolor}{\ignorespaces#2\unskip}\else\ignorespaces#2\unskip\fi}
\newcounter{changeversion}\newcounter{outerchangeversion}
\def\bc#1{\ifnum\thechangeversion>0\setcounter{outerchangeversion}{\thechangeversion}\setcounter{changeversion}{#1}\else\setcounter{changeversion}{#1}\fi\ifnum#1>\highlightchangessince\color{changecolor}\fi\ignorespaces}
\def\ec#1{\unskip\ifnum\theouterchangeversion>0\setcounter{changeversion}{\theouterchangeversion}\setcounter{outerchangeversion}{0}\else\setcounter{changeversion}{0}\fi\ifnum\thechangeversion>\highlightchangessince\else\normalcolor\fi}

\theoremstyle{plain}
\newtheorem{thm}{Theorem}
\newtheorem{lem}[thm]{Lemma}

\theoremstyle{definition}

\newtheorem*{rem}{Remark}

\renewcommand{\epsilon}{\varepsilon}

\newcommand{\reals}{\mathbb{R}}

\newcommand{\given}{\;\middle\vert\;}

\DeclareMathOperator{\inversegamma}{InverseGamma}
\DeclareMathOperator{\inversegaussian}{InverseGaussian}

\DeclareMathOperator{\tv}{TV}

\DeclareMathOperator{\tr}{tr}
\DeclareMathOperator{\diag}{Diag}

\DeclareMathOperator{\ind}{ind\,}

\DeclareMathOperator{\normal}{N}
\def\T{{ \mathrm{\scriptscriptstyle T} }}

\begin{document}

\title{Scalable Bayesian shrinkage and uncertainty quantification for high-dimensional 
regression}
\author{Bala Rajaratnam\thanks{Authors supported in part by the US National Science 
Foundation under grants DMS-CMG-1025465, AGS-1003823, DMS-1106642, and
DMS-CAREER-1352656, and by the US~Air Force Office of Scientific Research grant 
award FA9550-13-1-0043.}\\Department of Statistics, University of California, Davis\\ 
\\Doug Sparks\footnotemark[1]\hspace{.2cm}\\Department of Statistics, Stanford 
University\\ \\Kshitij Khare\thanks{Author supported in part by NSF grant 
DMS-15-11945.}\\Department of Statistics, University of Florida\\ 
\\Liyuan Zhang\\Department of Statistics, University of Florida}

\maketitle

\vspace*{-1.6em}

%
%
\begin{abstract}
Bayesian shrinkage methods have generated a lot of recent interest as tools for 
high-dimensional regression and model selection. These methods naturally facilitate 
tractable uncertainty quantification and incorporation of prior information. This benefit 
has led to extensive use of the Bayesian shrinkage methods across diverse 
applications. A common feature of these models, including the Bayesian lasso, 
global-local shrinkage priors, spike-and-slab priors is that the corresponding priors on 
the regression coefficients can be expressed as scale mixture of normals. 
While the three-step Gibbs sampler used to sample from the often intractable 
associated posterior density has been shown to be geometrically ergodic for several 
of these models \citep{khare2013, Pal:Khare:2014}, it has been demonstrated recently 
that convergence of this sampler can still be quite slow in modern high-dimensional 
settings despite this apparent theoretical safeguard. In order to address this 
challenge, we propose a new method to draw from the same posterior via a tractable 
two-step blocked Gibbs sampler. We demonstrate that our proposed two-step blocked 
sampler exhibits vastly superior convergence behavior compared to the original three-
step sampler in high-dimensional regimes on both real and simulated data. We also 
provide a detailed theoretical underpinning to the new method in the context of the 
Bayesian lasso. First, we prove that the proposed two-step sampler is geometrically 
ergodic, and derive explicit upper bounds for the (geometric) rate of convergence. 
Furthermore, we demonstrate theoretically that while the original Bayesian lasso chain 
is not Hilbert--Schmidt, the proposed chain is trace class (and hence Hilbert--Schmidt). 
The trace class property has useful theoretical and practical implications. It implies 
that the corresponding Markov operator is compact, and its (countably many) 
eigenvalues are summable. It also facilitates a rigorous comparison of the two-step 
blocked chain with ``sandwich" algorithms which aim to improve performance of the 
two-step chain by inserting an inexpensive extra step. 

\end{abstract}

\onehalfspacing

\section{Introduction}

\noindent
In modern statistics, high-dimensional datasets, where the number of 
covariates/features is more than the number of samples, are very common. Penalized 
likelihood methods such as the lasso \citep{tibshirani1996} and its variants 
simultaneously inducing shrinkage and sparsity in the estimation of regression 
coefficients. These goals are especially desirable when the number of coefficients to 
be estimated is greater than the sample size. One drawback of these penalized 
method is that it is not immediately obvious how to provide meaningful uncertainty 
quantification for the coefficient estimates. An alternative solution is to pursue a 
Bayesian approach by using shrinkage priors - priors that shrink the coefficients 
towards zero, by either having a discrete point mass component at zero (e.g., 
spike and slab priors \cite{george1993}) or a continuous density with a peak 
at zero. The uncertainty of the resulting estimates can be quantified in a natural way 
through the usual Bayesian framework (e.g., credible intervals). 

In fact, one can interpret the lasso objective function (or some monotone 
transformation thereof) as the posterior under a certain Bayesian model with an 
independent Laplace prior on the coefficients, as was noted immediately by 
\citeauthor{tibshirani1996}, and developed further by Park and Casella \cite{park2008} 
into the Bayesian lasso approach. Following the Bayesian lasso, a rich and interesting 
class of ``continuous global-local shrinkage priors'' has been developed in recent 
years. 
\setcitestyle{notesep={; }}%
\citep[see, for example,][and the references therein]
{armagan2013s,carvalho2010,griffin2010}.%
\setcitestyle{notesep={, }}%
The priors are typically scale mixtures of normals, and have a peak at zero to promote 
shrinkage. 

However, for most of these methods, the resulting intractable posterior is not 
adequately tractable to permit the closed-form evaluation of integrals. 
To address this problem, the respective authors have typically proposed a three-block 
Gibbs sampler based on a hierarchical formulation of the prior structure. This 
structure, which is essentially a type of data augmentation, leads to a tractable three-
step Gibbs sampler (one step for $\bm{\beta}$, $\sigma^2$ and the augmented 
parameter block each) that can be used to draw from the desired posterior. These 
posterior samples can then be used to construct credible intervals or other quantities 
of interest. 

Bayesian shrinkage methods (including the Bayesian lasso) have been extensively 
used in applications as diverse as genetics, finance, ecology, image processing, 
neuroscience, and clinical trials, receiving over 1,000 citations in total 
\citep[see, for example,][]
{yi2008,deloscampos2009,demiguel2009,jacquemin2014,xing2012,mishchenko2012,
gu2013, PongWong:2014, PongWong:Woolliams:2014} example). Nevertheless, the 
three-step Gibbs sampling schemes needed to implement these methods in practice 
can require considerable computational resources. Such computational concerns have 
primarily limited the application of these approaches to problems of only moderately 
high dimension (e.g., in the tens or hundreds of variables). This is often a serious 
shortcoming since for many modern applications, the number of variables $p$ is 
in the thousands if not more. Thus, methods to analyze or improve the efficiency of the 
Bayesian lasso algorithm, in terms of either computational complexity or convergence 
properties, are an important topic of research for modern high-dimensional settings. 

\citet{khare2013} proved that the three-step Gibbs sampler of \citeauthor{park2008} 
(which works for arbitrary $n$ and~$p$) is geometrically ergodic for arbitrary values of 
$n$ and~$p$ and they provide a quantitative upper bound for the geometric rate 
constant. Geometric ergodicity of the three-step samplers for the 
Normal-Gamma and Dirichlet-Laplace shrinkage models (\cite{griffin2010} and 
\cite{bhattacharya2015}, respectively) was established in \cite{Pal:Khare:2014}. 
However, it has been demonstrated that the geometric rate constant does indeed tend 
to~$1$ if $p/n\to\infty$ \citep{rajaratnam2015}. Thus, despite the apparent theoretical 
safeguard of geometric ergodicity, these three-step Gibbs samplers can take 
arbitrarily long to converge (to within a given total variation distance of the true 
posterior) if $p/n$ is large enough. This fact is problematic since the so-called 
``small~$n$, large~$p$'' setting is precisely where the use of the lasso and other 
regularized regression methods is most beneficial and hence most commonly 
espoused. 

Since the convergence properties of the original three-step Bayesian lasso Gibbs 
sampler deteriorate in high-dimensional regimes, it may be asked whether there exist 
alternative schemes for sampling from the same posterior that maintain a reasonably 
fast (i.e., small) geometric convergence rate when $p$ is large compared to~$n$.
Two commonly employed approaches to constructing such alternative MCMC 
schemes within the Gibbs sampling context are known as \emph{collapsing} and 
\emph{blocking}.  In a collapsed Gibbs sampler, one or more parameters are 
integrated out of the joint posterior, and a Gibbs sampler is constructed on the 
posterior of the remaining parameters.  Although a collapsed Gibbs sampler 
converges at least as fast as its uncollapsed counterpart \citep{liu1994w}, the resulting 
distributions may not be adequately tractable to permit the implementation of such a 
scheme in practice. In a blocked Gibbs sampler (also called a grouped Gibbs 
sampler), multiple parameters are combined and sampled simultaneously in a single 
step of the cycle that forms each iteration. It is generally understood that blocking 
usually improves the convergence rate with a careful choice of which parameters to 
group into the same step \citep{liu1994w}. 

In this paper, we propose a two-step/two-block Gibbs sampler in which the regression 
coefficients $\bm\beta$ and the residual variance~$\sigma^2$ are drawn in the same 
step of the Gibbs sampling cycle. This method turns out to be just as tractable as the 
original three-step/three-block Gibbs samplers. Indeed, the distributional forms of our 
proposed sampler coincide with the original chains, differing only in the shape and 
scale of the inverse gamma distribution from which~$\sigma^2$ is drawn.
Also, unlike the original three-step/three-block Gibbs chain, the two-step/two-block 
chain is reversible. We demonstrate empirically that in regimes where $p$ is much 
larger than~$n$, the convergence rate of the proposed two-block Gibbs samplers 
are vastly superior to those of the original three-block schemes. 

Next, we undertake a rigorous investigation into the theoretical properties of the 
proposed two-block chain in the context of the Bayesian lasso. We first establish 
geometric ergodicity for the blocked chain, and obtain an explicit upper bound for the 
rate of convergence. Geometric ergodicity (along with other mild regularity conditions) 
implies the existence of a Markov chain CLT, which allows users to provide standard 
errors for Markov chain based estimates of posterior quantities. 

We also prove that the (non--self-adjoint) Markov operator associated with the original 
Gibbs sampling Markov chain is not Hilbert--Schmidt (eigenvalues of the absolute 
value of this operator are not square-summable) for all $n$ and~$p$. On the other 
hand, we prove that the (positive, self-adjoint) Markov operator associated with the 
proposed blocked chain is trace-class (eigenvalues are in fact summable, and hence
square-summable). Note that all the aforementioned eigenvalues are less than 
$1$ in absolute value. These results indicate that the proposed Markov chain is more 
efficient than the original three-step Bayesian lasso chain. The blocked Markov 
chain can also be regarded as a Data Augmentation algorithm. Sandwich algorithms, 
which aim to improve the performance of DA algorithms by inserting an inexpensive 
extra step between the two steps of the DA algorithm, have gained a lot of attention in 
recent years (see 
\cite{Liu:Wu:1999, vanDyk:Meng:1999, Hobert:Marchev:2008, Khare:Hobert:2011} 
and the references therein). The trace class property for the blocked chain, along with 
results from \cite{Khare:Hobert:2011}, guarantees that a large variety of 
sandwich Markov chains are themselves trace class, and their eigenvalues are 
dominated by the corresponding eigenvalues of the blocked chain (with at least one 
strict domination). Thus, our trace class result provides theoretical support for the use 
of sandwich Markov chains to further improve speed and efficiency for sampling 
from the desired Bayesian lasso posterior distribution. 

It is well-understood (see for example \cite{jones2001}) that proving geometric 
ergodicity and trace class/Hilbert Schmidt properties for practical Markov chains such 
as these can be rather tricky and {\it a matter of art}, where each Markov chain 
requires a genuinely unique analysis based on its specific structure. Hence, a 
theoretical study of the two-step and three-step Markov chains (similar to the study 
undertaken in this paper for the Bayesian lasso) is a big and challenging undertaking, 
and is the task of ongoing and future research. 

The remainder of the paper is organized as follows. In Section~\ref{sec:original}, we 
revisit the original three-step chain. We propose the two-step chain in 
Section~\ref{sec:blocked}. Section~\ref{sec:numerical} provides a 
numerical comparison of the original and the proposed versions for both Bayesian 
lasso and the spike and slab priors (as representatives from the classes of continuous 
and discrete shrinkage priors) on simulated data, while Section~\ref{sec:applications} 
does the same for real data. We then focus on the theoretical properties 
of the original and proposed methods in the context of the Bayesian lasso model, and 
undertake a rigorous investigation in Section~\ref{sec:ge}.

\section{Bayesian shrinkage for regression} \label{sec:original}

Consider the model
\basn
\bm Y\mid\bm\beta,\sigma^2\sim \normal_n\left(\mu\bm1_n+\bm X\bm\beta,\,\sigma^2\bm I_n\right),
\label{model}
\easn
where $\bm Y\in\reals^n$ is a response vector, $\bm X$ is a known $n\times p$ 
design matrix of standardized covariates, $\bm\beta\in\reals^p$ is an unknown vector 
of regression coefficients, $\sigma^2>0$ is an unknown residual variance, and $\mu\in
\reals$ is an unknown intercept. As mentioned previously, in modern applications, 
often the number of covariates $p$ is much larger than the sample size $n$. To obtain 
meaningful estimates and identify significant covariates in this challenging situation, 
a common and popular approach is to shrink the regression coefficients towards zero.  
In the Bayesian paradigm, this can be achieved by using spike and slab priors, which  
are mixtures of a normal density with a spike at zero (low variance) and another 
normal density which is flat near zero (high variance), see 
\cite{mitchell1988, george1993}. A popular and more computationally efficient 
alternative is to use ``continuous" shrinkage prior densities that have a peak at zero, 
and tails approaching zero at an appropriate rate, see \cite{park2008, carvalho2010, 
polson2010, kyung2010, griffin2010, ADL:2013, polson2013} and the references 
therein. 

A common thread that runs through all the above shrinkage priors for high-dimensional 
regression is that they all can be expressed as scale mixtures of normal densities. In 
particular, the prior density for these models can be specified as 
\basn
\bm\beta\mid\sigma^2,\bm\tau\sim\normal_p\left(\bm0_p,\,\sigma^2\bm D_{\bm\tau}\right),
\qquad
\bm\tau\sim\pi(\bm\tau),
\label{mixture-general}
\easn

\noindent
where $\pi(\bm\tau)$ is a prior on~$\bm\tau=(\tau_1,\ldots,\tau_p)$. Suppose further 
that the prior on~$\sigma^2$ and~$\mu$ is once again the improper 
prior $\pi(\sigma^2,\mu)=1/\sigma^2$ and that this prior is independent of the prior on~
$\bm\tau$. After combining this prior structure with the basic regression model 
in~(\ref{model}) and integrating out~$\mu$, the remaining full conditional distributions 
are
\basn
\left.\bm\tau\given\bm\beta,\sigma^2,\bm Y\right.&\sim\pi\left(\bm\tau\given\bm\beta,\sigma^2,\bm Y\right),\notag\\
\left.\sigma^2\given\bm\beta,\bm\tau,\bm Y\right.&\sim\inversegamma\left[(n+p-1)/2,\,\|\tilde{\bm Y}-\bm X\bm\beta\|_2^2/2+
\bm\beta^\T
\bm D_{\bm\tau}^{-1}
\bm\beta
/2
\right],\label{gibbs-general}\\
\left.\bm\beta\given\sigma^2,\bm\tau,\bm Y\right.&\sim\normal_p\left(\bm A_{\bm\tau}
^{-1}\bm X^\T\tilde{\bm Y},\,\sigma^2\bm A_{\bm\tau}^{-1}\right),\notag
\easn
where $\bm A_{\bm\tau}=\bm X^\T\bm X+\bm D_{\bm\tau}^{-1}$ and $D_{\bm{\tau}} 
= Diag(\tau_1, \tau_2, \cdots, \tau_p)$. If it is feasible to draw from $\pi(\bm\tau\mid
\bm\beta,\sigma^2,\bm Y)$, then the three conditionals above may be used to 
construct a useful three-step Gibbs sampler to draw from the joint posterior 
$\pi(\bm\beta,\sigma^2 \mid \bm Y)$. The one-step transition density 
$\hat{k}$ with respect to Lebesgue measure on $\mathbb{R}^p \times \mathbb{R}_+$
given by
\begin{equation} \label{threebayes}
\hat{k}\left[\left(\bm{\beta}_0,\sigma^{2}_0\right), \left(\bm{\beta}_1,\sigma^2_1\right)\right] = \int_{\mathbb{R}_{+}^p} 
\pi \left(\sigma^2_1 \mid \bm{\beta}_1,\bm{\tau},\bm Y\right)\, \pi \left(\bm{\beta}_1\mid \bm{\tau},\sigma^{2}_0,\bm Y
\right)\, \pi \left( \bm{\tau} \mid \bm{\beta}_0,\sigma^{2}_0,\bm Y\right)\; d\bm{\tau}. 
\end{equation}

\noindent
Many commonly used Bayesian methods for high-dimensional regression can be 
characterized in the form of the priors in~(\ref{mixture-general}) and the Gibbs sampler 
in~(\ref{gibbs-general}). We now present a few such examples. 

\subsection{%
Spike-and-Slab Prior
} \label{sec:spikeandslab}

Now suppose instead that the $\tau_j$ are assigned independent discrete priors that each assign probability~$w_j$ to the point~$\kappa_j\zeta_j$ and probability $1-w_j$ to the point~$\zeta_j$, where $\zeta_j>0$ is small, $\kappa_j>0$ is large, and $0<w_j<1$.
This
formulation
is a slight modification of the prior proposed by \citet{george1993} to approximate the spike-and-slab prior of \citet{mitchell1988}.
Then the conditional posterior distribution of~$\bm\tau\mid\bm\beta,\sigma^2,\bm Y$ is a product of independent discrete distributions that each assign probability $\tilde w_j$ to the point~$\kappa_j\zeta_j$ and probability $1-\tilde w_j$ to the point~$\zeta_j$, where
\bas
\tilde w_j=\left\{1+\frac{(1-w_j)\sqrt{\kappa_j}}{w_j}\exp\left[-\frac{\beta_j^2}{2\sigma^2}\left(\frac{\kappa_j-1}{\kappa_j\zeta_j}\right)\right]\right\}^{-1}
\eas
by straightforward
modification
of the results of \citet{george1993}. 

\subsection{The Bayesian lasso} \label{sec:blasso}

\noindent
Suppose that the $\tau_j$'s are assigned independent $Exponential(\lambda^2/2)$ 
priors. It follows that the marginal prior of $\bm{\beta}$ (given $\sigma^2$) assigns 
independent Laplace densities to each component. In fact, the posterior mode in 
this setting is precisely the lasso estimate of \cite{tibshirani1996}. Hence, 
Park and Casella \cite{park2008} refer to this approach as the Bayesian lasso. 
In this setting, the conditional posterior distribution of~$\bm\tau\mid\bm\beta,
\sigma^2,\bm Y$ assigns independent inverse Gaussian distributions to 
each $1/\tau_j$, and hence is easy to sample from. 

Following the three-step Gibbs sampler of \citeauthor{park2008}, two useful
alternative Gibbs samplers for sampling from the Bayesian lasso posterior were 
proposed by \citet{hans2009}. However, both samplers, namely the standard and the 
orthogonal sampler, require a sample size~$n$ at least as large as the number of 
variables~$p$, since they require the design matrix to have full column rank.
The standard sampler can perform poorly when the predictors are highly correlated, 
while the orthogonal sampler becomes very computationally expensive as $p$ grows 
(recall that $p$ still must be less than~$n$). The $n>p$ assumption can thus be
a serious limitation, as it precisely excludes the high-dimensional settings targeted by 
the Bayesian lasso. 

\subsection{Global-local continuous shrinkage priors} \label{sec:globallocal}

\noindent
Suppose we assume that $\tau_j = \eta \lambda_j$ for $j = 1,2, \cdots, p$. 
Here, $\eta$ controls global shrinkage towards the origin while the local scale 
parameters $\lambda_j$ allow component wise deviations in shrinkage. Hence, these 
priors are called as global-local shrinkage priors. A variety of global-local shrinkage 
priors have been proposed in the Bayesian literature, 
see \cite{polson2010, bhattacharya2015} for an exhaustive list.  For almost all of these 
models, sampling from the posterior distribution is performed by using a three-block 
Gibbs sampler, with ${\boldsymbol \beta}$ and $\sigma^2$ being two of the blocks, 
and all the shrinkage parameters (local and global) being the third block. As an 
example, we consider the Dirichlet-Laplace shrinkage prior from 
\cite{bhattacharya2015}. The regression version of this model was considered in 
\cite{Pal:Khare:2014}, and is provided as follows. 
\begin{eqnarray} \label{DirichletLaplaceModel}
& & {\bf y} \mid {\boldsymbol \beta},\sigma^2 \sim N(X{\boldsymbol \beta},\sigma^2 
I_n) \nonumber\\
& &{\boldsymbol \beta} \mid \sigma^2, {\boldsymbol \psi}, {\boldsymbol \phi}, \theta 
\sim N(0, \sigma^2D_{\boldsymbol \tau}) \mbox{ where } \tau_j = \psi_j \phi_j^2 
\theta^2 \nonumber\\
& & \sigma^2 \sim \text{Inverse-Gamma}(\alpha,\xi) \quad \alpha,\xi>0 \text{   fixed} \nonumber\\
& & \psi_1, \psi_2, \cdots, \psi_p \stackrel{i.i.d.}{\sim} Exp\left( \frac{1}{2} \right) \nonumber\\
& & (\phi_1,\phi_2,...,\phi_p) \sim Dir(a,a,...,a), \theta \sim \text{Gamma} \left( pa,\frac{1}{2} \right). 
\end{eqnarray} 

\noindent
Here ${\boldsymbol \eta} = ({\boldsymbol \psi}, {\boldsymbol \phi}, \theta)$ denotes 
the collection of all shrinkage parameters, and $a$ is a known positive constant. It 
can be seen that the full conditional distributions of ${\boldsymbol \beta}$ and 
$\sigma^2$ are exactly the same as in (\ref{gibbs-general}). Based on results in 
\cite{bhattacharya2015}, it is shown in \cite{Pal:Khare:2014} that samples from the full 
conditional distribution of ${\boldsymbol \eta}$ can be generated by making 
draws from a bunch of appropriate Generalized Inverse Gaussian densities. 

\subsection{%
Student's~$t$
Prior
} \label{sec:t}

First, suppose that the $\tau_j$ are assigned independent inverse-gamma priors with shape parameter $\nu_j/2$ and scale parameter $\eta_j/2$.
This structure corresponds (by integrating out $\bm\tau$) to specificiation of the prior on $\bm\beta\mid\sigma^2$ as a product of independent scaled
Student's~$t$
priors, where the prior on each $\beta_j\mid\sigma^2$ has $\nu_j$ degrees of freedom and scale parameter $(\eta_j\sigma^2)^{1/2}$.   Such
Student's~$t$
priors are a popular choice when a mildly informative prior is desired \citep[see, e.g.,][]{gelman2013}.
In this case, the conditional posterior distribution of $\bm\tau\mid\bm\beta,\sigma^2,\bm Y$ is a product of independent inverse-gamma distributions, where the distribution of each $\tau_j\mid\bm\beta,\sigma^2,\bm Y$ has shape parameter $(\nu_j+1)/2$ and scale parameter $(\eta_j+\beta_j^2/\sigma^2)/2$.

\subsection{%
Bayesian Elastic Net
} \label{sec:elasticnet}

As a final example, suppose instead that the $\tau_j$ are assigned independent 
continuous priors, each with density
\bas
\pi(\tau_j)=\frac{\lambda_1}{2(1-\lambda_2\tau_j)^2}\exp\left[-\frac{\lambda_1\tau_j}{2(1-\lambda_2\tau_j)}\right]
\eas
with respect to Lebesgue measure on the interval $(0,\lambda_2^{-1})$, where $\lambda_1,\lambda_2>0$.
This structure corresponds (by integrating out $\bm\tau$) to specification of the prior on $\bm\beta\mid\sigma^2$ as a product of independent priors, where the prior on each $\beta_j\mid\sigma^2$ has density with respect to Lebesgue measure that is proportional to
\bas
\pi(\beta_j\mid\sigma^2)\propto\exp\left(-\sqrt{\frac{\lambda_1}{\sigma^2}}\,\left|\beta_j^{}\right|-\frac{\lambda_2}{2\sigma^2}\beta_j^2\right).
\eas
Then the conditional posterior distribution of $\bm\tau\mid\bm\beta,\sigma^2,\bm Y$ is
\bas
\left.\left(\frac1{\tau_j}-\lambda_2\right)\given\bm\beta,\sigma^2,\bm Y\right.\sim\ind\inversegaussian\left(\sqrt{\frac{\lambda_1\sigma^2}{\beta_j^2}},\,\lambda_1\right).
\eas
This prior structure and corresponding Gibbs sampler are known as the Bayesian elastic net \citep{li2010,kyung2010}.

\section{The fast Bayesian shrinkage algorithm} \label{sec:blocked}

\noindent
While the three-step Gibbs sampler (with transition density $\hat{k}$) provides a 
useful and straightforward way to sample from the posterior density, its slow 
convergence in high-dimensional settings with $p >> n$ is discussed by 
\citet{rajaratnam2015}. In particular, it is demonstrated that the slow convergence 
problem in these settings arises due to the high a~posteriori dependence between 
$\bm\beta$ and~$\sigma^2$. It was argued by \citet{liu1994w} that the convergence 
rate of Gibbs samplers can often be improved by grouping highly dependent 
parameters together into a single block of a blocked Gibbs sampler. Of course, 
if the conditional densities for the blocked Gibbs sampler are not easy to sample 
from, then any possible gain from blocking is likely to be overshadowed by the extra 
effort in sampling from such densities. 

In an effort to address the slow convergence of the three-step Gibbs sampler in 
high-dimensional settings, we prove the following lemma. Its proof may be found in the 
Supplementary Material.
\begin{lem} \label{lem:sigma-given-tau}
For the Bayesian model in (\ref{mixture-general}), $\sigma^2\mid\bm\tau,\bm Y$ has 
the inverse gamma distribution with shape parameter $(n-1)/2$ and scale parameter $
\tilde{\bm Y}^\T(\bm I_n-\bm X\bm A_{\bm\tau}^{-1}\bm X^\T)\tilde{\bm Y}/2$.
\end{lem}

\noindent
This lemma facilitates the construction of a novel blocked two-step Gibbs sampler 
which can be used to generate samples from the joint posterior density of 
$(\bm{\beta}, \sigma^2)$, and which is as tractable as the original three-step Gibbs 
sampler. This blocked Gibbs sampler alternates between drawing $(\bm\beta,
\sigma^2)\mid\bm\tau$ and $\bm\tau\mid(\bm\beta,\sigma^2)$. In particular,
$(\bm\beta,\sigma^2)\mid\bm\tau$ may be drawn by first drawing $\sigma^2\mid\bm
\tau$ and then drawing $\bm\beta \mid \sigma^2,\bm\tau$. In other words, the 
blocked Gibbs sampler may be constructed by replacing the draw of $\sigma^2\mid
\bm\beta,\bm\tau,\bm Y$ in~(\ref{gibbs-general}) with a draw of $\sigma^2 \mid 
\bm{\tau}$ as given by Lemma~\ref{lem:sigma-given-tau}. In particular, the blocked 
Gibbs sampler cycles through drawing from the distributions
\basn
&\left.\bm\tau\given\bm\beta,\sigma^2,\bm Y\right. \sim\pi\left(\bm\tau\given\bm\beta,
\sigma^2,\bm Y\right) \label{gibbs-blocked}\\
&\left.(\bm\beta, \sigma^2)\given\bm\tau,\bm Y\right.\sim
\left\{
\begin{aligned}
\left.\sigma^2\given\bm\tau,\bm Y\right.&\sim\inversegamma\left[(n-1)/2,\,\tilde{\bm Y}^\T\left(\bm I_n-\bm X\bm A_{\bm\tau}^{-1}\bm X^\T\right)\tilde{\bm Y}/2\right],\\
\left.\bm\beta\given\sigma^2,\bm\tau,\bm Y\right.&\sim N_p\left(\bm A_{\bm\tau}^{-1}\bm X^\T\tilde{\bm Y},\,\sigma^2\bm A_{\bm\tau}^{-1}\right), 
\end{aligned}
\right.\notag
\easn
and has one-step transition density $k$ with respect to
Lebesgue
measure on $\mathbb{R}^p \times 
\mathbb{R}_+$
given
by 
\begin{equation} \label{eqtwobayes}
k\left[\left(\bm{\beta}_0, \sigma^2_0\right), \left(\bm{\beta}_1, \sigma^2_1\right)\right] = \int_{\mathbb{R}_{+}^p} \pi (\bm{\tau} \mid \bm{\beta}_0, \sigma^2_0, \bm 
Y)\, \pi (\sigma^2_1 \mid \bm{\tau}, \bm Y)\, \pi(\bm{\beta}_1 \mid \sigma^2_1, \bm{\tau}, \bm Y)\; d\bm{\tau}. 
\end{equation}

\noindent
Note that the blocked Gibbs sampler is just as tractable as the original sampler since 
the only the parameters of the inverse gamma distribution are modified. In particular, 
although the blocked Gibbs sampler requires inversion of the matrix 
$\bm A_{\bm\tau}$ to draw $\sigma^2$, this inversion must be carried out anyway to 
draw~$\bm\beta$, so no real increase in computation is required. 

Note that our Markov chains move on the $(\bm{\beta}, \sigma^2)$ space as these are 
the parameters of interest. In other words, we integrate out the augmented parameter 
$\bm{\tau}$ when we compute our transition densities $\hat{k}$ and $k$ in 
(\ref{threebayes}) and (\ref{eqtwobayes}) respectively. If we consider the 
``lifted versions" of our chains where we move on the $(\bm{\beta}, \sigma^2, 
\bm{\tau})$ space (do not integrate out $\bm{\tau}$ for the transition densities), then 
the results in \cite{liu1994w} imply that the Markov operator corresponding to the 
two-step chain has a smaller operator norm as compared to the three-step chain. 
However, no comparisons of the spectral radius, which dictates the rate of 
convergence of the non-reversible lifted chains, is available. Furthermore, no 
theoretical comparison of the actual Markov chains (with transition densities $\hat{k}$ 
and $k$) is available. Hence, in subsequent sections, we undertake a numerical 
comparison (on both simulated and real data) as well as a theoretical comparison of 
the proposed two-step sampler and the original three-step sampler. 

\begin{rem}
Note that in the case of many global-local shrinkage priors such as 
in Section \ref{sec:globallocal}, ${\boldsymbol \tau}$ is a function of the global and 
local shrinkage parameters. Hence, to sample from ${\boldsymbol \tau}$, one needs to 
sample from the conditional distribution of the collection of shrinkage parameters 
(denoted by ${\boldsymbol \eta}$ in Section \ref{sec:globallocal}) given 
${\boldsymbol \beta}$ and $\sigma^2$. However, the distribution of $({\boldsymbol 
\beta}, \sigma^2)$ depends on ${\boldsymbol \eta}$ only through ${\boldsymbol \tau}$, 
and samples from the joint density of $({\boldsymbol \beta}, \sigma^2)$ given 
${\boldsymbol \eta}$ can be generated exactly as described in 
(\ref{gibbs-blocked}). 
\end{rem}

\begin{rem}
We use the term \emph{blocked} rather than \emph{grouped} to avoid confusion with 
existing methods such as the \emph{group Bayesian lasso} \citep{kyung2010}, in 
which the notion of grouping is unrelated to the concept considered here. Note also
that the original three-step sampler could already be considered a blocked Gibbs 
sampler since the $\beta_j$ and $\tau_j$ are not drawn individually.  Thus, our use of 
the term \emph{blocked} in describing the proposed two-step sampler should be 
understood as meaning that the Gibbs sampling cycle is divided into fewer blocks than 
in the original Gibbs sampler.
\end{rem}

\section{Numerical Comparison} \label{sec:numerical}

\noindent
Note that as illustrated in Section \ref{sec:original} our approach is applicable to a wide 
variety of Bayesian high-dimensional regression methods. In this section, we will 
undertaken a numerical comparison between the convergence rates and efficiency of 
the original and proposed chains for two such methods. One is the spike-and-slab 
approach from \ref{sec:spikeandslab} (representing methods with point mass or 
discrete shrinkage priors), and the second one is Bayesian lasso approach from 
Section \ref{sec:blasso} (representing continuous shrinkage priors).  As a proxy for the 
actual rate of convergence of the chain, we consider the autocorrelation in the
marginal~$\sigma^2_k$ chain. Note that this autocorrelation is exactly equal to the 
geometric convergence rate in the simplified case of standard Bayesian regression 
\citep{rajaratnam2015} and is a lower bound for the true geometric convergence rate 
in general \citep{liu1994w}. 

\subsection{Numerical comparsion for Bayesian lasso}

\noindent
Figure~\ref{fig:ac-p} plots the autocorrelation for the original three-step and proposed 
two-step Bayesian lasso Gibbs samplers versus~$p$ for various values of~$n$. 
(Similar plots under varying sparsity and multicollinearity may be found in the 
Supplementary Material). The left side of Figure~\ref{fig:dacf} plots the same 
autocorrelation versus $\log(p/n)$ for a wider variety of values of $n$ and~$p$. It is 
apparent that this autocorrelation for the proposed two-step Bayesian lasso is 
bounded away from~$1$ for all $n$ and~$p$. The center and right side of 
Figure~\ref{fig:dacf} show dimensional autocorrelation function (DACF) surface plots 
\citep[see][]{rajaratnam2015} for the original (left) and  (right) Bayesian lasso Gibbs 
samplers. (See the Supplementary Material for details of the generation of the various 
numerical quantities that were used in the execution of these chains). It is clear that 
the autocorrelation tends to~$1$ as $p/n\to\infty$ for the original Bayesian lasso but 
remains bounded away from~$1$ for the the proposed two-step Bayesian lasso. 

\begin{figure}[htbp]
\centering
\includegraphics[scale=0.38]{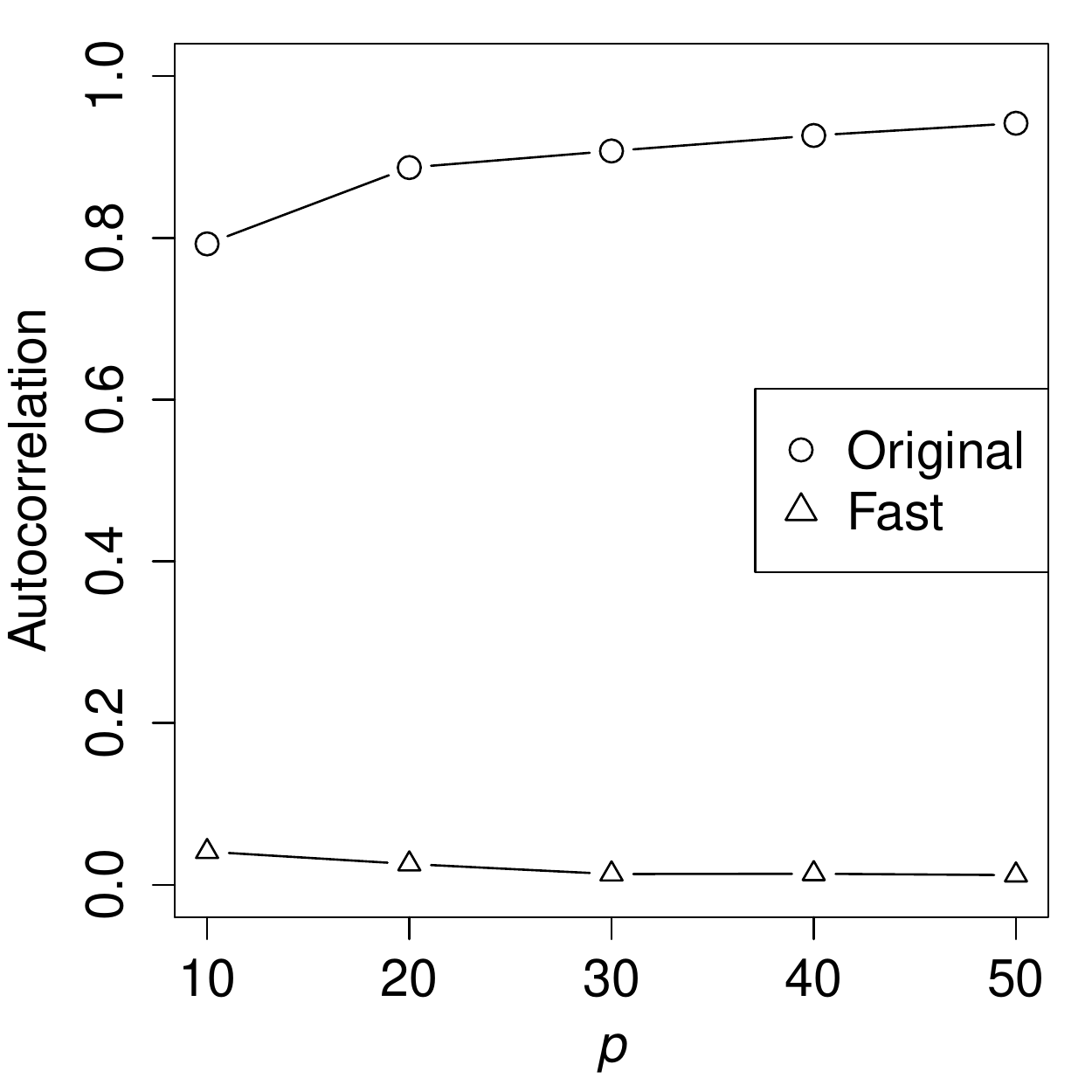}
\includegraphics[scale=0.38]{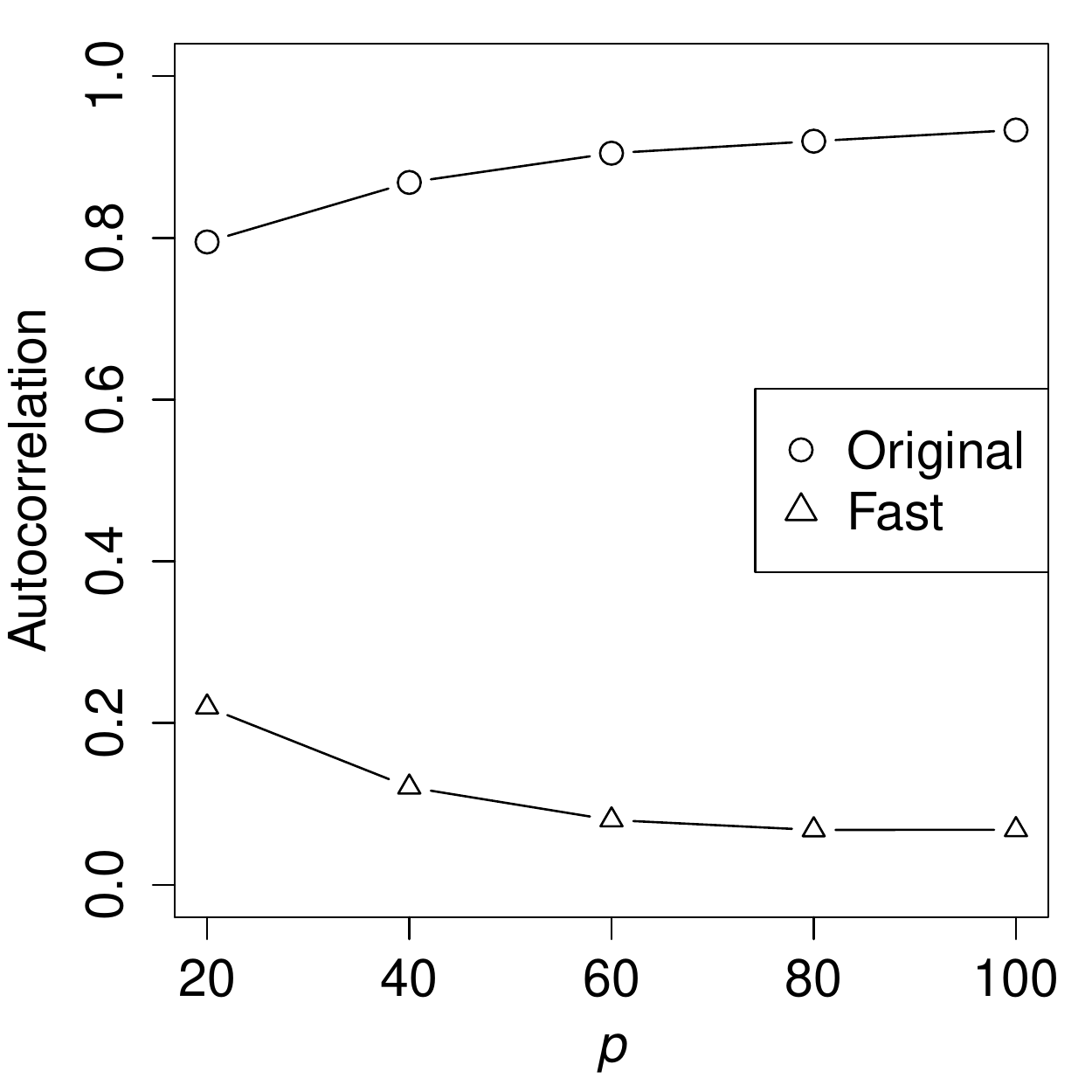}
\includegraphics[scale=0.38]{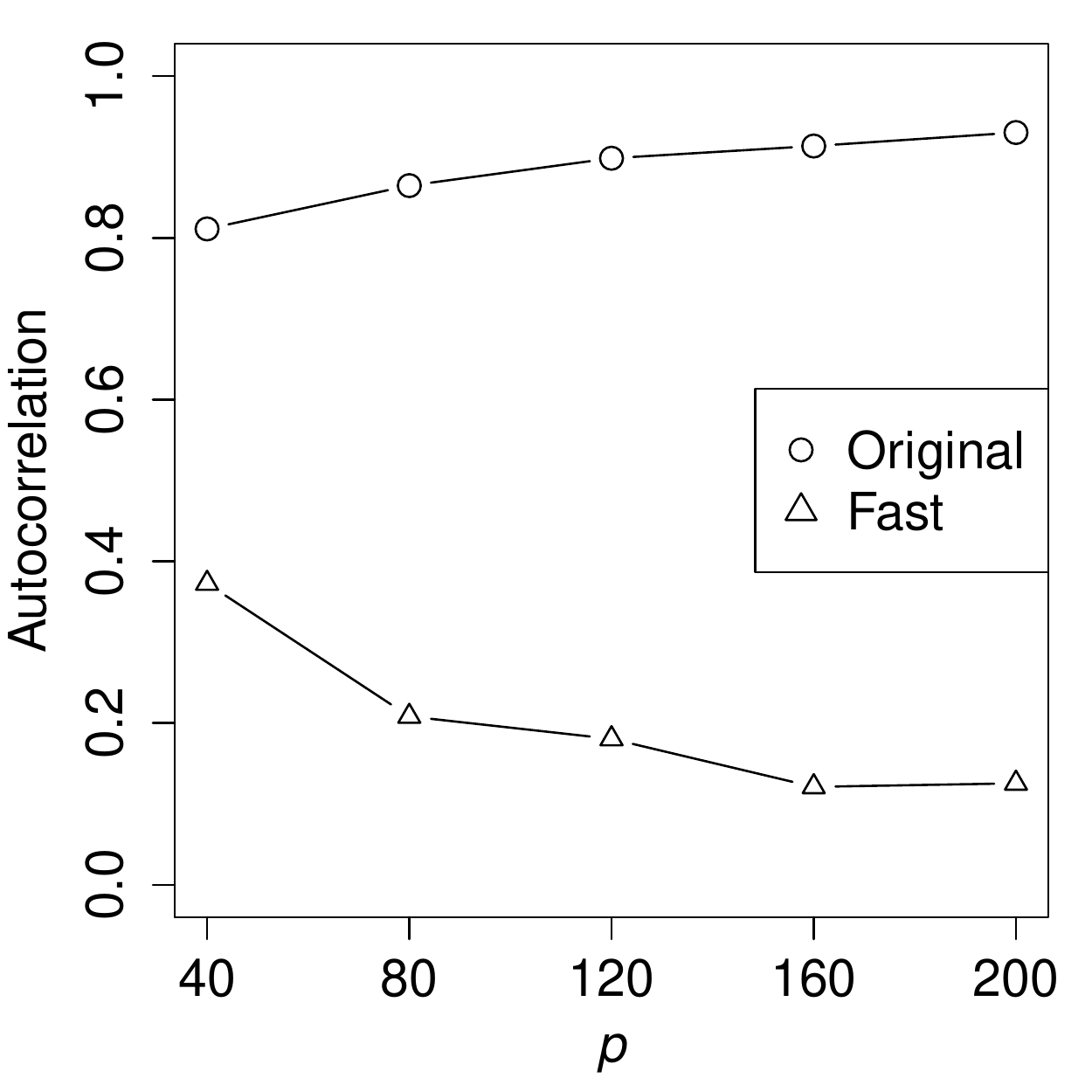}
\vspace*{-1em}
\caption{%
Autocorrelation of $\sigma^2_k$ versus $p$ for the original three-step and proposed 
two-step Bayesian lasso Gibbs samplers with $n=5$ (left), $n=10$ (center), and 
$n=20$ (right).
See the Supplementary Material for details of the generation of the
numerical quantities used
in the execution of these chains.
}
\label{fig:ac-p}
\end{figure}

\begin{figure}[htbp]
\centering
\includegraphics[scale=0.38]{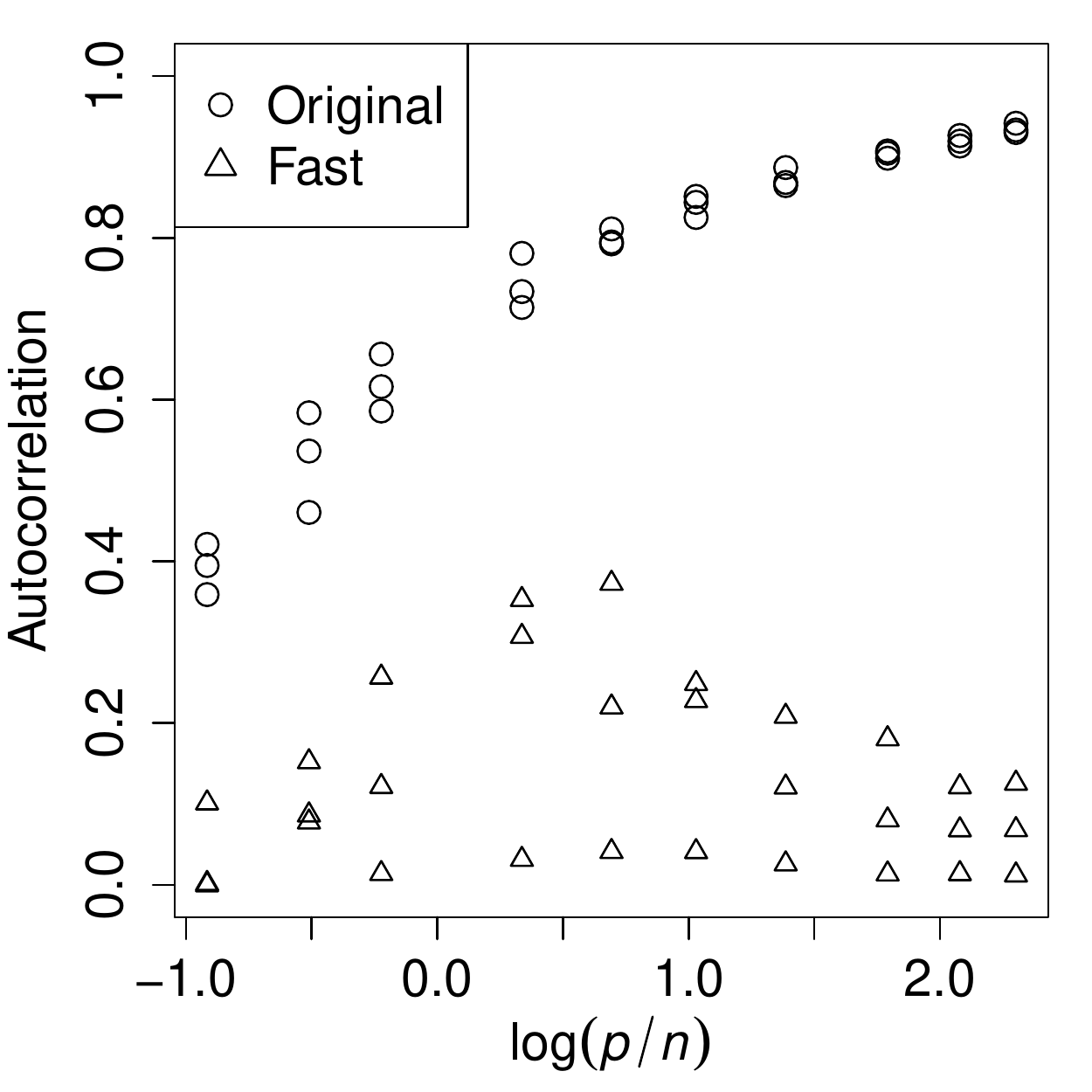}
\includegraphics[scale=0.38]{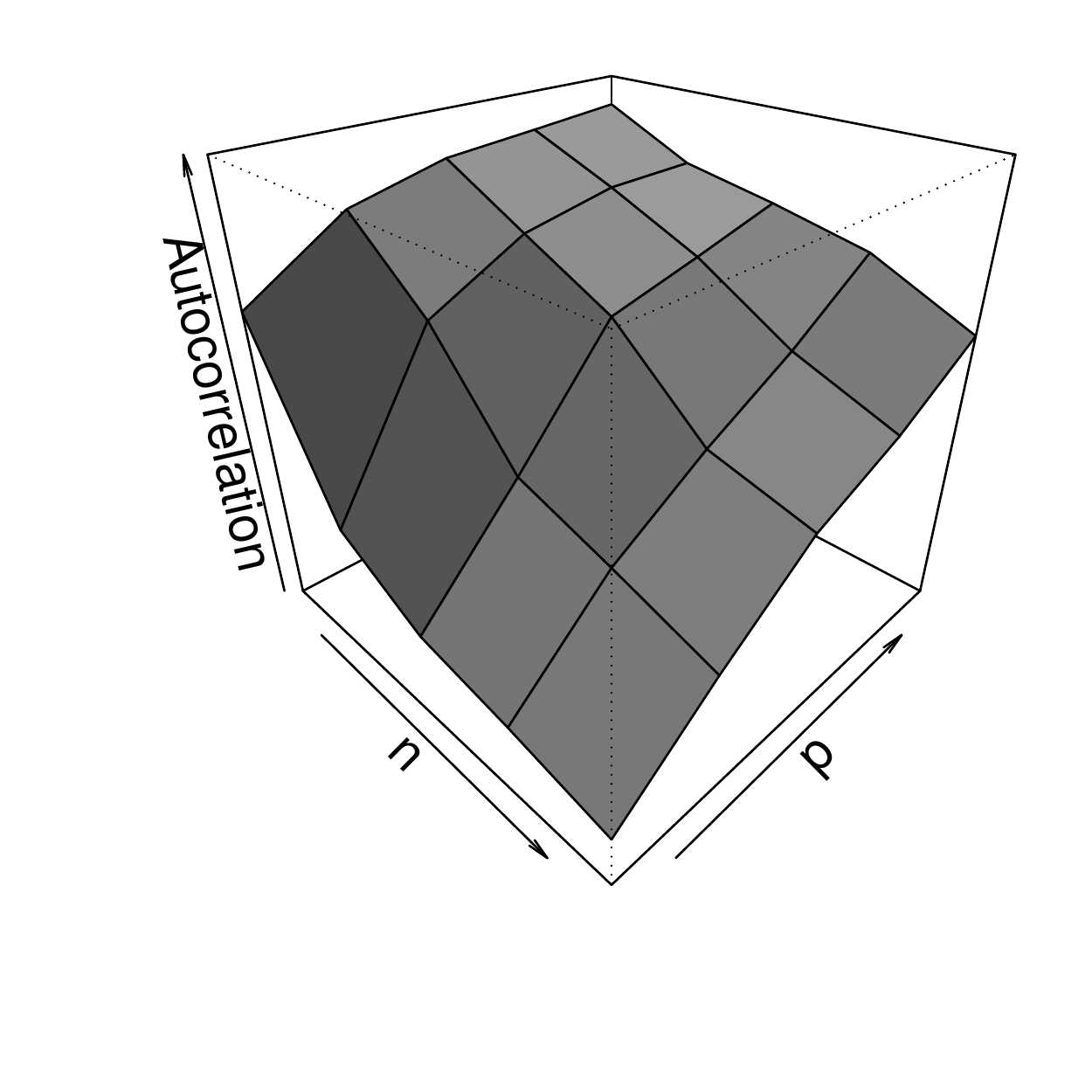}
\includegraphics[scale=0.38]{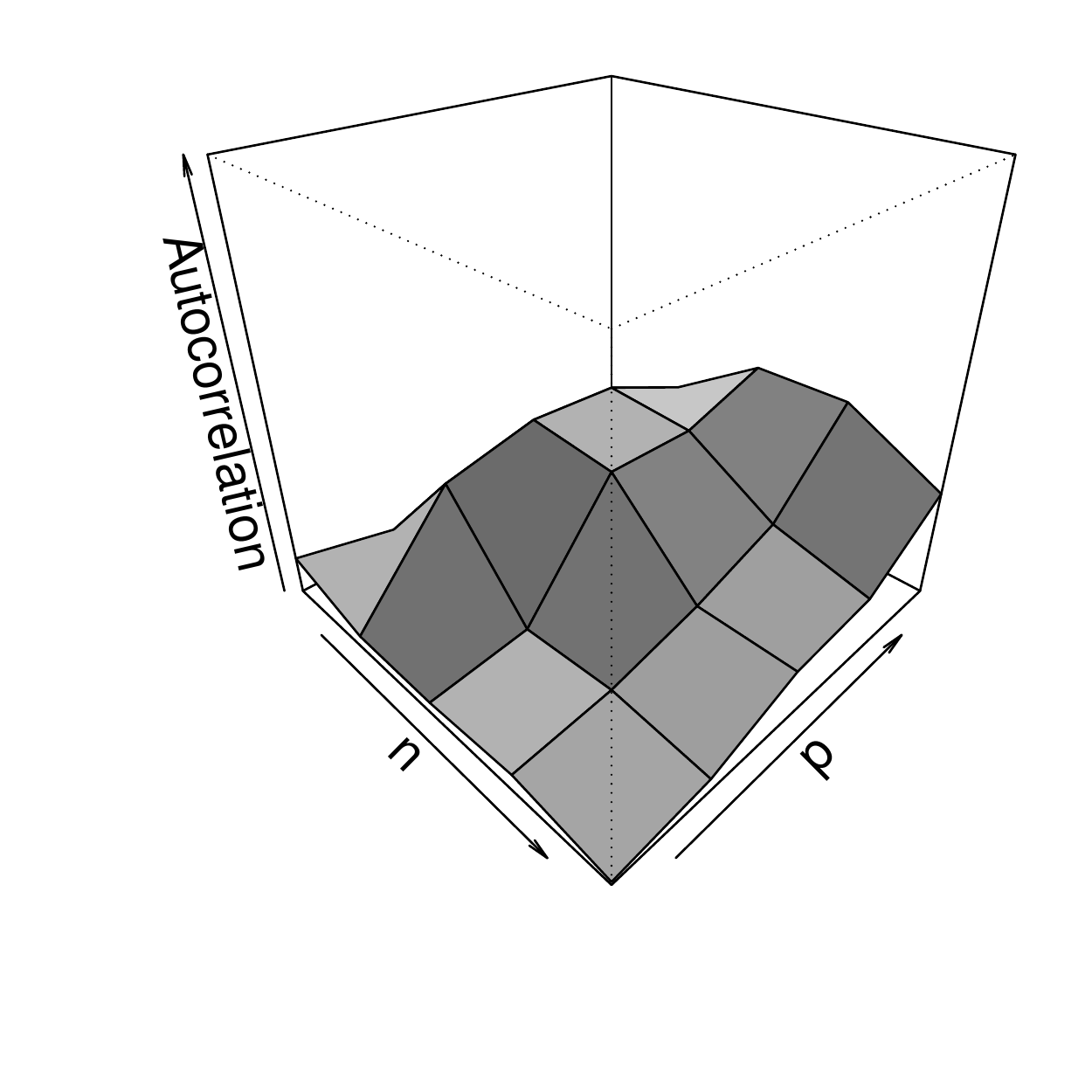}
\vspace*{-1em}
\caption{%
Autocorrelation of $\sigma^2_k$ versus $\log(p/n)$ for the original three-step and 
proposed two-step Bayesian lasso Gibbs samplers with various values of $n$ 
and~$p$ (left).Dimensional autocorrelation function (DACF) surface plots for the 
$\sigma^2_k$ chain relative to $n$ and~$p$ for the original (center) and two-step 
(right) Bayesian lasso Gibbs samplers.
See the Supplementary Material for details of the generation of the
numerical quantities used
in the execution of these chains.
}
\label{fig:dacf}
\end{figure}

\subsection{Numerical comparison for spike and slab}

We also applied both spike-and-slab Gibbs samplers to simulated data that was 
constructed in a manner identical to the simulated data for the Bayesian lasso in 
Figure~\ref{fig:ac-p}.  The resulting autocorrelations of the $\sigma^2_k$ chains are 
shown in Figure~\ref{fig:ac-ss-p}.  As with the Bayesian lasso, it is clear that the 
two-step blocked version of the spike-and-slab sampler exhibits dramatically improved 
convergence behavior in terms of $n$ and~$p$ as compared to the original three-step 
sampler. 
\begin{figure}[htbp]
\centering
\includegraphics[scale=0.38]{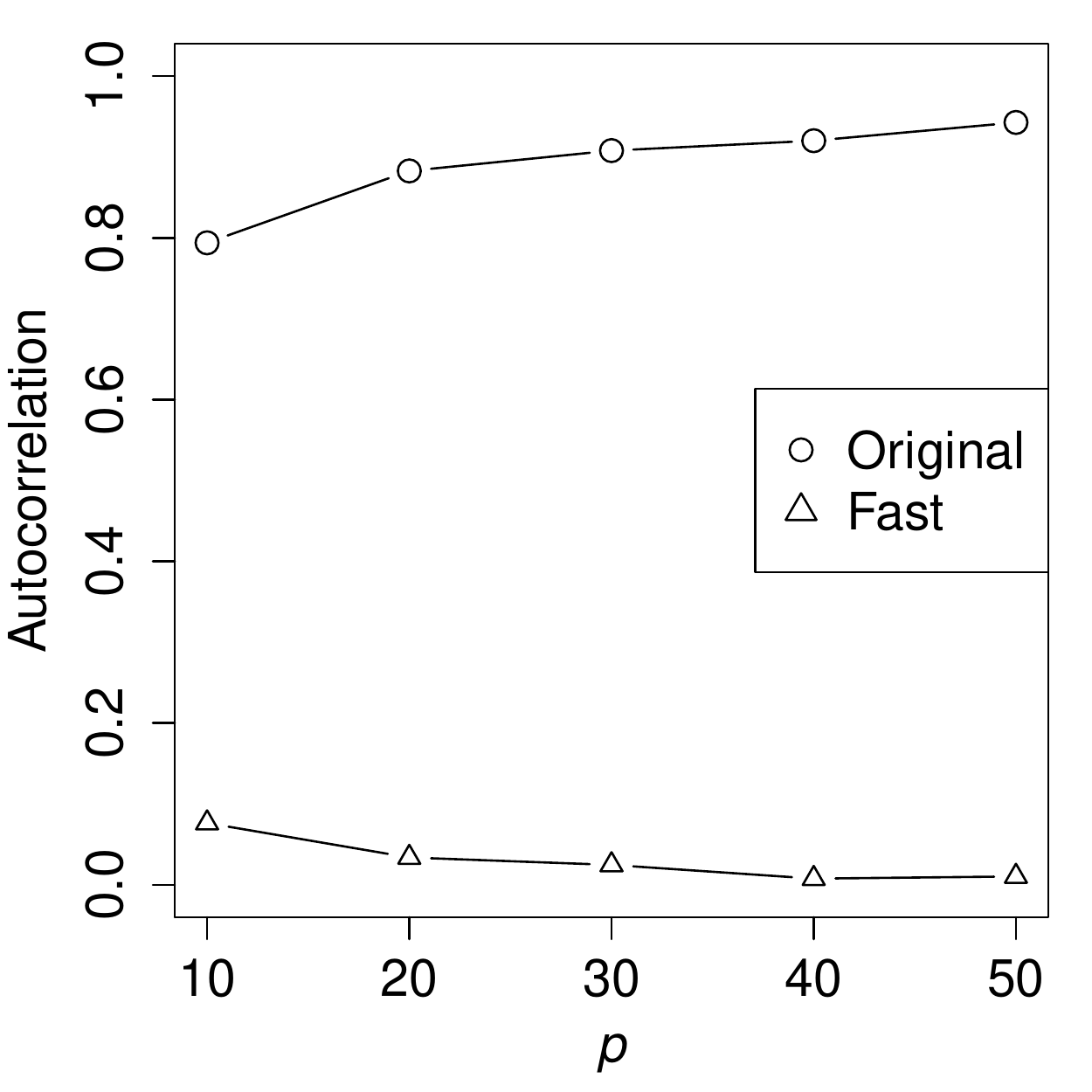}
\includegraphics[scale=0.38]{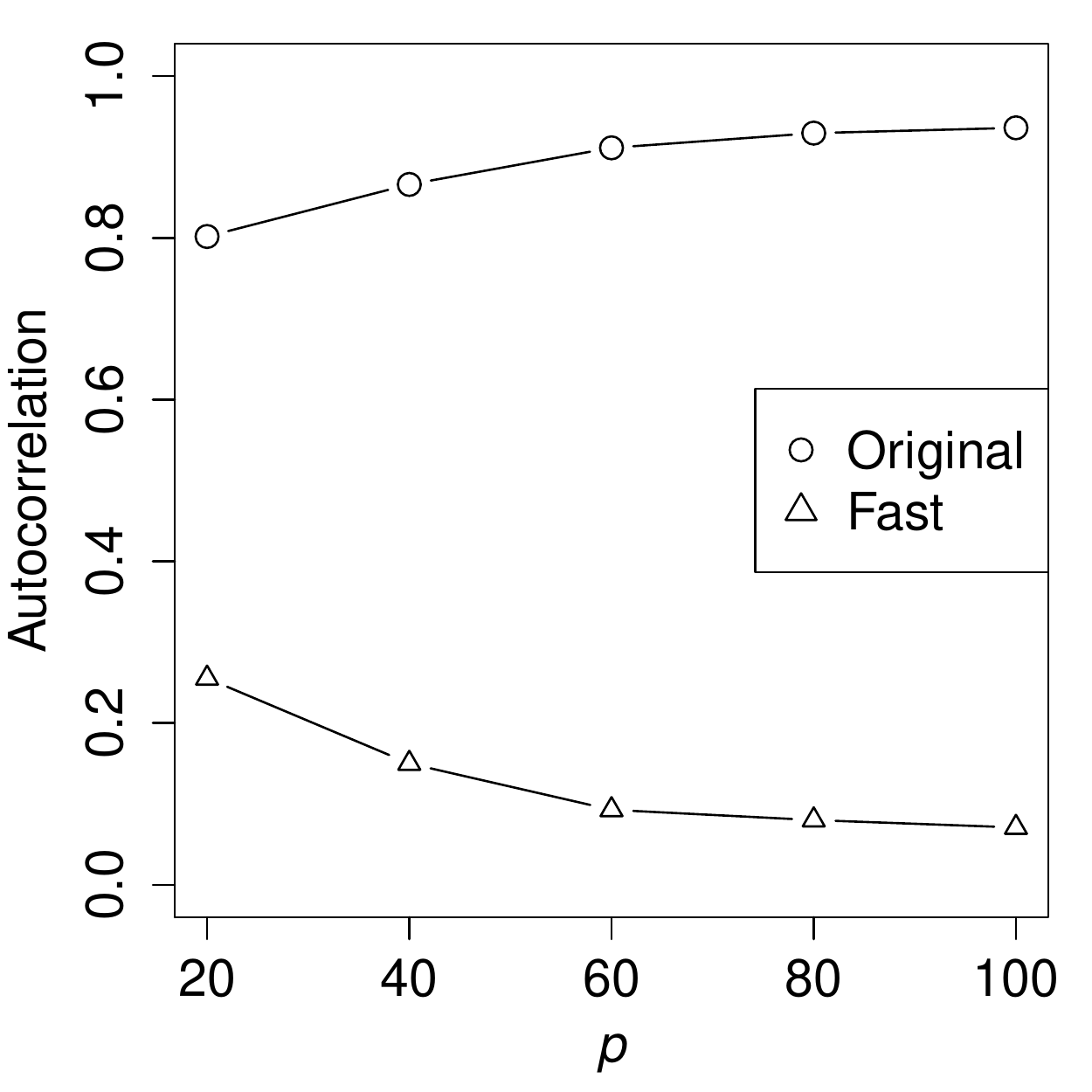}
\includegraphics[scale=0.38]{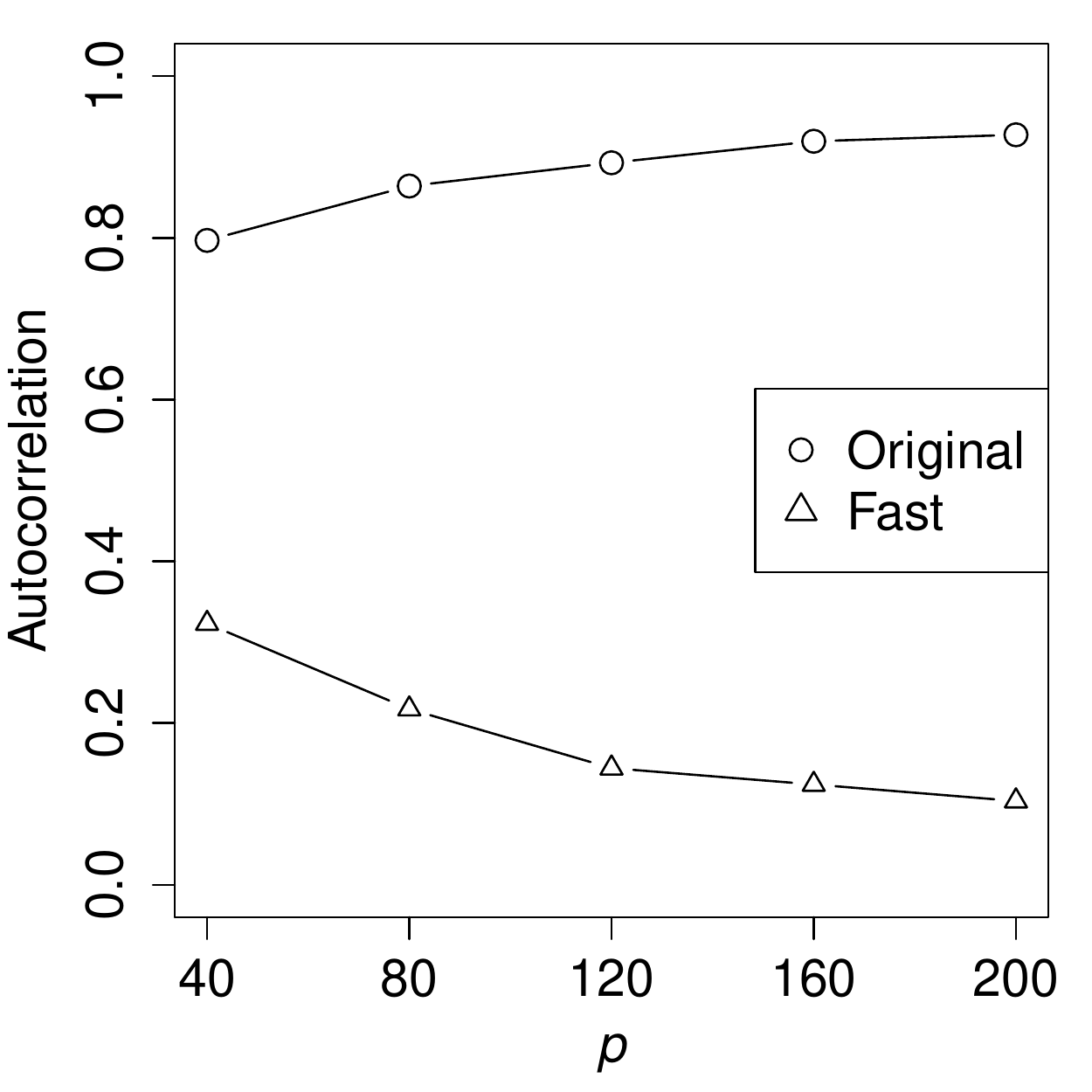}
\vspace*{-1em}
\caption{%
Autocorrelation of $\sigma^2_k$ versus $p$ for the original two-step and proposed 
three-step spike-and-slab Gibbs samplers with $n=5$ (left), $n=10$ (center), and 
$n=20$ (right).
See the Supplementary Material for details of the generation of the
numerical quantities used
in the execution of these chains.
}
\label{fig:ac-ss-p}
\end{figure}

%
%

\section{Applications to Real Data} \label{sec:applications}

\noindent
We now illustrate the benefits of the proposed two-step approach through its 
application to three readily available regression data sets.

\subsection{Gene Expression Data} \label{subsec:gene}

\subsubsection{Results for Bayesian lasso}

\noindent
We first evaluate the performance of the original and proposed chains using the 
Bayesian lasso model through their application to the gene expression data of \citet{scheetz2006}, which is publicly available as the data set \texttt{eyedata} in the R~package \texttt{flare} \citep{Rflare}.  This data set includes $n=120$ observations of a response variable and $p=200$ covariates.
The columns of~$\bm X$ were standardized to have mean zero and squared Euclidean norm~$n$.
The regularization parameter was taken to be $\lambda=0.2185$
so that the number of nonzero coefficients in the lasso estimate is $\min\{n,p\}/2=60$.
Both the oriiginal three-step and proposed two-step Bayesian lassos were executed for 10,000 iterations
(after discarding an initial ``burn-in'' period of 1,000 iterations),
which took
about
100~seconds
on a 2.6~GHz machine for each chain.  Further
numerical details (such as the initialization of the chains) may be found by direct inspection of the R code, which we provide in its entirety in the Supplementary Material.\nocite{R,Rlars,Rstatmod}

We now show that
the Markov chain of the proposed two-step Bayesian lasso mixes substantially better than that of the original Bayesian lasso when applied to this gene expression data set.
Specifically, the lag-one autocorrelation of the $\sigma^2$ chain
for the two-step Bayesian lasso is
0.3885,
whereas it is
0.7794
for the original Bayesian lasso.
We also computed the effective sample size of the $\sigma^2$ chain using the \texttt{effectiveSize} function of the R~package \texttt{coda} \citep{Rcoda}.  The effective sample size is
4,160
for the two-step Bayesian lasso and
1,240
for the original Bayesian lasso.
These effective sample sizes essentially demonstrate that the two-step Bayesian lasso 
needs only about
30\;\!\%
as many Gibbs sampling iterations to obtain an MCMC sample of equivalent quality to that of the original Bayesian lasso.
Since the computational time for a single iteration is essentially the same for both algorithms, the two-step Bayesian lasso also requires only about
30\;\!\%
as much time as the original Bayesian lasso to obtain an MCMC sample of equivalent quality.
Such considerations are particularly important when considering the high-dimensional data sets to which the lasso and Bayesian lasso are often applied.

\begin{rem}
The reader may wonder why we have chosen to examine the autocorrelation and effective sample size of the $\sigma^2$ chain instead of the $\bm\beta$ chain (i.e., the individual $\beta_j$ chains).
The reason is simply that the autocorrelation and effective sample size of the $\bm\beta$ chain can fail to accurately convey the mixing behavior of regression-type Gibbs samplers.
In short, the sampled iterates of certain quadratic forms in~$\bm\beta$ can be very highly autocorrelated even though the $\beta_j$ chains themselves have an autocorrelation near zero.
Such convergence behavior leads to serious inaccuracy when using the MCMC output for uncertainty quantification, which is the very reason why the Bayesian lasso is employed in the first place.
In contrast, there are strong theoretical connections between the autocorrelation of the $\sigma^2$ chain and the true convergence rate for regression-type Gibbs samplers.
See \citet{rajaratnam2015} for a thorough discussion of these concepts.
\end{rem}

\subsubsection{Results for Spike-and-Slab} \label{subsec:numerical-ss}

\noindent
Next, we investigate the performance of the two-step Gibbs sampler in~(\ref{gibbs-blocked}) relative to the original 
sampler in~(\ref{gibbs-general}) using the spike-and-slab prior. Both the original and proposed Gibbs samplers 
were executed for 10,000 iterations (after discarding an initial ``burn-in'' period of 1,000 iterations), which took 
about 80~seconds on a 2.6~GHz machine for each chain.  Other numerical details may be found by direct 
inspection of the R~code, which we provide in its entirety in the Supplemental Material. 

As was the case with the two-step Bayesian lasso, the blocked version of the 
spike-and-slab Gibbs sampler exhibits mixing properties that are substantially better 
than those of the corresponding original sampler.  Specifically, the lag-one 
autocorrelation of the $\sigma^2$ chain for the two-step spike-and-slab Gibbs sampler 
is 0.0187 whereas it is 0.5174 for the original spike-and-slab Gibbs sampler.  
Moreover, the effective sample size of the $\sigma^2$ chain is 9,372 for the fast 
spike-and-slab Gibbs sampler and 2,977 for the original spike-and-slab Gibbs 
sampler. Thus, the blocking approach of the proposed two-step sampler once again 
yields a sampler that is dramatically more efficient than the original three-step sampler. 

\subsection{Infrared Spectroscopy Data} \label{subsec:spectroscopy}

We next evaluate the performance of the original three-step and proposed two-step 
approaches in the context of Bayesian lasso through their application to the infrared 
spectroscopy data of \citet{osborne1984} and \citet{brown2001}, which is publicly 
available as the data set \texttt{cookie} in the R package \texttt{ppls} \citep{Rppls}. 
We used the first $n=40$ observations and the first of the response variables 
(corresponding to fat content).  The $p=700$ covariates of this data set exhibit high 
multicollinearity, with a median pairwise correlation of~$0.96$. The columns of 
$\bm X$ were standardized to have mean zero and squared Euclidean norm~$n$.  
The regularization parameter was taken to be $\lambda=0.0504$, which yields $
\min\{n,p\}/2=20$ nonzero coefficients in the lasso estimate. Both the original and 
two-step Bayesian lassos were executed for 10,000 iterations (after discarding an 
initial ``burn-in'' period of 1,000 iterations), which took about 25~minutes on a 
2.6~GHz machine for each chain. The R~code is provided in the Supplementary 
Material. \nocite{R}

two-stepAs with the gene expression data of the previous subsection, the Markov 
chain of the two-step Bayesian lasso mixes much better than that of the original 
Bayesian lasso when applied to this infrared spectroscopy data set. The lag-one 
autocorrelation of the $\sigma^2$ chain is 0.0924 for the fast Bayesian lasso and 
0.9560 for the original Bayesian lasso. The effective sample size of the $\sigma^2$ 
chain is 7,790 for the fast Bayesian lasso and 225 for the original Bayesian lasso. 
two-stepThus, the two-step Bayesian lasso requires
only about
3\;\!\%
as many Gibbs sampling iterations
to achieve the same effective sample size as the original Bayesian lasso.
Hence, the two-step Bayesian lasso requires only about
3\;\!\%
as much time as the original Bayesian lasso to obtain an MCMC sample of equivalent quality.
(Also see the previous subsection for further comments
on the
consideration of the $\sigma^2$ chain specifically.)

\subsection{Communities and Crime Data}
\label{subsec:crime}

For a third application to real data, we consider the communities and crime data of \citet{redmond2002}, which is publicly available through the UC Irvine Machine Learning Repository \citep{lichman2013} as the Communities and Crime Data Set.
We chose 50 covariates from that data set and constructed $p=1325$ covariates by taking the first-order, quadratic, and multiplicative interaction terms involving those 50 raw covariates.
We again illustrate the performance of the original three-step and proposed two-step 
chains un the context of the Bayesian lasso model in a small-$n$, large-$p$ setting. 
We choose $n=10$ of the 1994 total observations uniformly at random.
Each of the 1325 columns (each of length~10) of the resulting $\bm X$ matrix were 
standardized to have mean zero and squared Euclidean norm~$n$.  The 
regularization parameter was taken to be $\lambda=1.331$, which yields 
$\min\{n,p\}/2=5$ nonzero coefficients in the lasso estimate for the particular 
observations selected. Both the original and two-step Bayesian lassos were executed for 
10,000 iterations (after discarding an initial ``burn-in'' period of 1,000 iterations), which 
took about 90~minutes on a 2.6~GHz machine for each chain.  The R code is 
provided in the Supplementary Material.

Once again, the Markov chain of the two-step Bayesian lasso mixes considerably faster than that of the original Bayesian lasso.
The lag-one autocorrelation of the $\sigma^2$ chain is
0.0017
for the two-step Bayesian lasso and
0.9942
for the original Bayesian lasso.
The effective sample size of the $\sigma^2$ chain is
10,000
for the two-step Bayesian lasso and
29
for the original Bayesian lasso.
(Note that the \texttt{effectiveSize} function involves the use of AIC to select the order of an autoregressive model, and obtaining an effective sample size exactly equal to the number of MCMC iterations simply indicates that AIC has selected a trivial autoregressive model of order zero.)
Thus, the two-step Bayesian lasso requires
only about 0.3\;\!\%
as many Gibbs sampling iterations to achieve the same effective sample size as the original Bayesian lasso.  Hence, the two-step Bayesian lasso requires
only about 0.3\;\!\%
as much time as the original Bayesian lasso to obtain an MCMC sample of equivalent quality.
(Also see the previous subsections for further comments
on the
consideration of the $\sigma^2$ chain specifically.)



\begin{table}[htbp]
\centering
\begin{tabular}{lrrccccc}
\toprule
&&&\multicolumn{2}{c}{Autocorrelation}&\multicolumn{3}{c}{Effective Size}\\
\cmidrule(r){4-5}\cmidrule(l){6-8}
\multicolumn{1}{c}{Data Set}&\multicolumn{1}{c}{$n$}&\multicolumn{1}{c}{$p$}&Two-step&Original&Two-step&Original&
Ratio
\\\midrule
Gene Expr. (Spike-Slab)&120&200&0.0187&0.5174&\phantom{0}9,372&2,977&
\phantom{00}3.15
\\
Gene Expr. (Bayesian lasso)&120&200&0.3885&0.7794&\phantom{0}4,160&1,240&
\phantom{00}3.4
\\
Infrared Spectroscopy&40&700&0.0924&0.9560&\phantom{0}7,790&\phantom{0,}225&
\phantom{0}34.6
\\
Communities \& Crime&10&1,325&0.0017&0.9942&10,000&\phantom{0,0}29&
344.8
\\
\bottomrule
\end{tabular}
\caption{Autocorrelation and effective sample size for the $\sigma^2$ chain of the two-step Bayesian lasso and original Bayesian lasso as applied to the three datasets in Section~\ref{sec:applications}.  Note that $n$ and~$p$ denote (respectively) the sample size and number of covariates  for each data set.}
\label{tab:applications}
\end{table}

\subsection{Summary of Applications} \label{subsec:applications-summary}

The results for all three data sets in this section are summarized in Table~
\ref{tab:applications}. It is clear that the two-step chain and the original three-step 
chain exhibit the same behavior (in regard to $n$ and~$p$) for these real data sets as 
for the simulated data sets of Section~\ref{sec:numerical}.  Specifically, the original 
three-step chain mixes very slowly when $p$ is large relative to~$n$, as illustrated by 
the large autocorrelation and small effective sample size of the $\sigma^2$ chain.
In contrast, the two-step chain enjoys excellent mixing properties in all $n$ and 
$p$ regimes and becomes better as $p$ gets larger (in fact, there is an unexpected 
``benefit of dimensionality'').

\section{Theoretical Properties}
\label{sec:ge}

\noindent
In this section we undertake an investigation into the theoretical convergence and 
spectral properties of the proposed two-step chain and the original three-step chains. 
In particular, we investigate whether these Markov chains converge at a geometric 
rate (geometric ergodicity) and also the specific behavior of their singular values 
(trace class/Hilbert Schmidt properties). As stated in the introduction, proving 
geometric ergodicity and trace class/Hilbert Schmidt properties for practical Markov 
chains such as these can be rather tricky and a matter of art, where each Markov 
chain requires a genuinely unique analysis based on its specific structure. Hence, we 
provide results here for the Bayesian lasso model outlined in Section \ref{sec:blasso}. 
A  theoretical study of all the other two-step and three-step Markov chains is a big and 
challenging undertaking, and is the task of ongoing and future research. 

We first proceed to establish geometric ergodicity of the two-step Bayesian lasso Gibbs sampler using the method of 
\citet{rosenthal1995}.
This approach provides a quantitative upper bound on the geometric convergence rate. 
To express this result rigorously, we first define some notation.
For every $m\ge1$, let $F_m(\sigma^2_0,\bm\beta_0)$ denote the distribution of the $m$th iterate for the two-step Bayesian lasso 
Gibbs sampler initialized at $(\bm\beta_0, \sigma^2_0)$. Let $F$ denote the stationary distribution of this Markov chain, i.e., the 
true joint posterior distribution of $(\bm\beta, \sigma^2)$. Finally, let $d_{\tv}$ denote total variation distance.
Then we have the following result, which we express using notation similar to that of Theorem~2 of \citet{jones2001}.
\begin{thm}
\label{thm:ge}
Let $0<\gamma<1$, $b=((n+3)p)(16 \gamma + (n + 3))/(64 \gamma)$, and $d > 2b/(1-\gamma)$. Then for any $0<r<1$,
\bas
d_{\tv}\left[F_k\left(\sigma^2_0,\bm\beta_0\right),F\right]\le(1-\epsilon)^{rk}+\left(\frac{U^r}{\alpha^{1-r}}\right)^k\left(1+\frac{b}{1-\gamma}+\frac{\lambda\left\|\bm\beta_0\right\|_2^2}{\sigma^2_0}\right),
\eas
for every $k\ge1$, where $\epsilon=\exp\left(-p \sqrt{d} \right)$, $U=1+2(\gamma d+b)$, and 
$\alpha=(1+d)/(1+2b+\gamma d)$. 
\end{thm}

The proof of Theorem~\ref{thm:ge} uses the lemma below, which is proven in the Supplementary Material.

\begin{lem}
\label{lem:ratio}
Let $C_{\bm\tau}=\tilde{\bm Y}^\T(\bm I_n-\bm X\bm A_{\bm\tau}^{-1}\bm X^\T)\tilde{\bm Y}$.
Then
$\|\bm A_{\bm\tau}^{-1}\bm X^\T\tilde{\bm Y}\|_2^2/C_{\bm\tau}\le\|\bm\tau\|_1/4$.
\end{lem}

With Lemma~\ref{lem:ratio} in place, we can now prove Theorem~\ref{thm:ge}.

\begin{proof}[Proof of Theorem~\ref{thm:ge}]
We appeal to Theorem~2 of \citet{jones2001}
\citep[see also Theorem~12 of][]{rosenthal1995}.
To apply this result, it is necessary to establish a drift condition and an associated minorization condition.
Let $V(\bm\beta,\sigma^2)=\lambda^2{
\bm\beta^\T\bm\beta
}/\sigma^2$.
Let $(\bm\beta_0, \sigma^2_0)$ denote the initial value for the blocked Markov chain,
and let
$(\bm\beta_1,\sigma^2_1)$
denote
the first iteration of the chain.
Then let
$E_0$ denote
expectation conditional on
$(\bm\beta_0, \sigma^2_0)$.
To establish the drift condition, observe that
\basn
E_0\left[V\left(\bm\beta_1, \sigma^2_1\right) \right]
&=\lambda^2\,E_0 \left[{
\sigma_1^{-2}
}\,E\left({
\bm\beta_1^\T\bm\beta_1
}\given\bm\tau_1,\sigma^2_1\right) \right]\notag\\
&=\lambda^2\,E_0 \left[\tr\left(\bm A_{\bm\tau}^{-1}\right)+{
\|\bm A_{\bm\tau}^{-1}\bm X^\T\tilde{\bm Y}\|_2^2
}\,E\left({
\sigma_1^{-2}
}\given\bm\tau \right)\right]\notag\\
&=\lambda^2\,E_0 \left[\tr\left(\bm A_{\bm\tau}^{-1}\right)+{
\|\bm A_{\bm\tau}^{-1}\bm X^\T\tilde{\bm Y}\|_2^2
(n-1)/C_{\bm\tau}
}\right]\notag\\
&\le\lambda^2\,E_0 \left[\tr\left(\bm D_{\bm\tau} \right)+{
\|\bm\tau\|_1(n-1)/4
}\right]
{}=E_0 \left(\|\bm\tau\|_1\right)\,(n+3)\lambda^2/4,
\label{drift-1}
\easn
where the inequality is
by Lemma~\ref{lem:ratio} and the fact that
$\tr(\bm A_{\bm\tau}^{-1})=\tr[(\bm X^\T\bm X+\bm D_{\bm\tau}^{-1})^{-1}]\le\tr[(\bm D_{\bm\tau}^{-1})^{-1}]=\tr(\bm D_{\bm\tau})$.
Then continuing from~(\ref{drift-1}), we have
\basn
E_0 \left[V\left(\bm\beta_1,\sigma^2_1\right)\right]
&
{}\le{}
\frac{(n+3)\lambda^2}4\sum_{j=1}^p\left(\frac{\left|\beta_{0,j}\right|}{\lambda{
\sigma_0
}}+\frac1{\lambda^2}\right)
=\frac{n+3}4\sum_{j=1}^p\frac{\lambda\left|\beta_{0,j}\right|}{
\sigma_0
}+\frac{(n+3)p}4\label{drift-2}
\easn
by the basic properties of the inverse Gaussian distribution.
Now note that $u\le\delta u^2+(4\delta)^{-1}$ for all $u\ge0$ and $\delta>0$.
Applying this result to~(\ref{drift-2}) for each~$j$ with $u=\lambda|\beta_{0,j}|/{
\sigma_0
}$ and $\delta=4\gamma/(n+3)$ yields
\bas
E_0 \left[V\left(\bm\beta_1,\sigma^2_1\right)\right]
&\le\gamma\sum_{j=1}^p\frac{\lambda^2\beta_{0,j}^2}{\sigma^2_0}+\frac{n+3}4\left[\frac{(n+3)p}{16\,\gamma}\right]+\frac{(n+3)p}4
=\gamma\,V\left( \bm\beta_0,\sigma^2_0 \right)+b,
\eas
establishing the drift condition.

To establish the associated minorization condition, observe that
$k$ as defined in~(\ref{eqtwobayes}) is given by
\basn
&
k\left[\left(\sigma^2_0,\bm\beta_0\right), \left(\bm\beta_1,\sigma^2_1\right)\right]
\notag\\
&\quad=
\int_{\mathbb{R}^p_+} \frac{\lambda^p}{(2\pi)^{p/2}}
\exp\left[\frac{\lambda\|\bm\beta_0\|_1}{\sigma_0}-\frac{\bm\beta_0^\T\bm D_{\bm\tau}^{-1}\bm\beta_0^\T}{2\sigma^2_0}-
\frac{\lambda^2\|\bm\tau\|_1}2\right]\prod_{j=1}^p\tau_j^{-1/2}\times\notag\\
&\quad\quad
\frac{(C_{\bm\tau}/2)^{(n-1)/2}}{\Gamma\left(\frac{n-1}2\right)\,{
\sigma^{n+1}_1
}}\exp\left(-\frac{C_{\bm\tau}} 
{2{
\sigma^2_1
}}\right) \; \frac{\det(\bm A_{\bm\tau})^{1/2}}{(2\pi{
\sigma^2_1
})^{p/2}}
\exp\left[\frac1{2{
\sigma^2_1
}}\left\|\bm A_{\bm\tau}^{1/2}({
\bm\beta_1
}-\bm A_{\bm\tau}^{-1}\bm X^\T\tilde{\bm Y})\right\|_2^2\right] 
\;d \bm\tau \notag\\
&\quad= \int_{
\mathbb{R}^p_+
}
\frac{(\lambda/2\pi)^p}{2^{(n-1)/2}\,\Gamma\left(\frac{n-1}2\right)}
\;
\frac{\det\left(\bm A_{\bm\tau}\bm D_{\bm\tau}^{-1}\right)^{1/2}\,C_{\bm\tau}^{(n-1)/2}}{
\sigma^{n+p+1}_1
}\exp\left(-
\frac{\lambda^2\left\|\bm\tau\right\|_1}2\right) \times \notag\\
&\qquad\qquad \exp\left(-\frac{\|\tilde{\bm Y}-\bm X{
\bm\beta_1
}\|_2^2}{2{
\sigma^2_1
}}-\frac{{
\bm\beta_1^\T
}\bm D_{\bm\tau}^{-1}{
\bm\beta_1
}}
{2{
\sigma^2_1
}}\right)
\exp\left(\frac{\lambda\left\|\bm\beta_0\right\|_1}{\sigma_0}-\frac{\bm\beta_0^\T\bm D_{\bm\tau}^{-1}\bm\beta_0}
{2\sigma^2_0}\right)\;d \bm\tau. 
\label{iteration}
\easn
Now suppose that $V(\bm\beta_0, \sigma^2_0)\le d$. Then $\lambda^2\beta_{0,j}^2/\sigma^2_0\le d$ for each~$j$.
Let $\bm\xi=d^{1/2}\lambda^{-1}\bm1_p$.
Then
\bas
\exp\left(\frac{\lambda\left\|\bm\beta_0\right\|_1}{\sigma_0}-\frac{\bm\beta_0^\T\bm D_{\bm\tau}^{-1}\bm\beta_0}{2\sigma^2_0}\right)
&\ge
\exp\left(-\frac{\bm\xi^\T\bm D_{\bm\tau}^{-1}\bm\xi}{2}\right)
=\epsilon\exp\left(\lambda\|\bm\xi\|_1-\frac{\bm\xi^\T\bm D_{\bm\tau}^{-1}\bm\xi}{2}\right),
\eas
noting that $\epsilon=\exp(-\lambda\|\bm\xi\|_1)$.
Combining this inequality with~(\ref{iteration}) yields
\bas
k\left[\left(\bm\beta_0, \sigma^2_0\right),{
\left(\bm\beta_1, \sigma^2_1\right)
}\right]
&\ge \int_{
\reals_+^p
}
\frac{(\lambda/2\pi)^p\det\left(\bm A_{\bm\tau}\bm D_{\bm\tau}^{-1}\right)^{1/2}\,C_{\bm\tau}^{(n-1)/2}}{2^{(n-1)/2}\,\Gamma\left(\frac{n-1}2\right)\,
\sigma^{n+p+1}_1
}\exp\left(-\frac{\lambda^2\left\|\bm\tau\right\|_1}2\right) \times\\
&\quad\exp\left(\!-\frac{\|\tilde{\bm Y}-\bm X{
\bm\beta_1
}\|_2^2}{2{
\sigma^2_1
}}-\frac{{
\bm\beta_1^\T
}\bm D_{\bm\tau}^{-1}{
\bm\beta_1
}}{2{
\sigma^2_1
}}\right)
\;\epsilon\exp\left(\!\lambda\|\bm\xi\|_1-\frac{\bm\xi^\T\bm D_{\bm\tau}^{-1}\bm\xi}{2}\right)\,d \bm\tau\\
&=\epsilon\;k\left[\left(\bm\xi, 1\right),{
\left(\bm\beta_1,\sigma^2_1\right)
}\right]
\eas
for all $(\bm\beta, \sigma^2) \in \mathbb{R}^p \times \mathbb{R}_+$. Thus, the minorization condition is established.
\end{proof}

Note that for Theorem~\ref{thm:ge} to be useful, the bound must actually decrease with~$k$.
Thus, it is necessary to choose $r$ small enough that $U^r/\alpha^{1-r}<1$.
Then for small enough~$r$, the bound is dominated by the term $(1-\epsilon)^{rk}$, which is approximately $(1-r\epsilon)^k$ for small $r$ and~$\epsilon$.
Now observe that $d>2b>n^2p/32$. It follows that 
\bas
1-r\epsilon=1-r\exp(-p\,d^{1/2})>1-r\exp(-n\,p^{3/2}/32^{1/2}),
\eas
\noindent
which tends to~$1$ exponentially fast as $n$ or $p$ tends to infinity.  Thus, although Theorem~\ref{thm:ge}
establishes that the two-step Bayesian lasso Gibbs sampler is geometrically ergodic and provides a bound for the rate constant,
it is not clear how sharp it is. Hence, the bound may not be particularly informative in high-dimensional contexts.

Additional theoretical insight into the differences between the original and blocked sampling schemes may be gained through consideration of the spectral properties of the two Gibbs samplers.
The following theorem provides a key difference between these two chains.

\begin{thm} \label{thmsptlpty}
The Markov operator corresponding to the original three-step Gibbs transition density $\hat{k}$ is not Hilbert-Schmidt, while 
the Markov operator corresponding to the blocked Gibbs transition density $k$ is trace class (and hence Hilbert-Schmidt) 
{for all} possible values of $p$ and~$n$. 
\end{thm}

The proof of this result is quite
technical
and long
and may be found in the Supplementary Material. As discussed in the introduction, 
the above result implies that the eigenvalues of the absolute value of the
non--self-adjoint
operator associated with the original 
Gibbs chain are not square-summable.
In contrast, it also implies that
the eigenvalues of the self-adjoint and positive operator corresponding to the 
blocked Gibbs chain are summable (which implies square-summability, as these 
eigenvalues are bounded by $1$). Based on this result, we would expect the blocked 
Gibbs chain to be more efficient than the original Gibbs chain. The numerical results 
(for high-dimensional regimes) provided in 
Sections \ref{sec:numerical} and \ref{sec:applications} confirm this assertion.

\newpage

\renewcommand{\thefigure}{S\arabic{figure}}
\renewcommand{\thetable}{S\arabic{table}}
\renewcommand{\theequation}{S\arabic{equation}}
\setcounter{section}{0}
\renewcommand{\thesection}{S\arabic{section}}

\section*{Supplementary Material}

\noindent
This supplementary material is organized as follows.  Section~\ref{sec:proof} contains proofs of Lemmas~\ref{lem:sigma-given-tau}~and~\ref{lem:ratio} and 
Theorem~\ref{thmsptlpty}.
Section~\ref{sec:more-numerical} contains additional numerical results.
Section~\ref{sec:numerical-details} contains details of the numerical results in
Sections~\ref{sec:numerical}~and~\ref{sec:applications}.
Section~\ref{sec:code} contains R code for executing both the original and two-step Bayesian lasso algorithms.
In addition, the R code for generating the numerical results of Sections~\ref{sec:numerical},~\ref{sec:applications},~and~\ref{sec:more-numerical} is provided in its entirety in an accompanying file
Supplementary material comprises proofs of Lemmas~\ref{lem:sigma-given-tau}~and~\ref{lem:ratio}
and
Theorem \ref{thmsptlpty}, additional numerical results, and
details of
the numerical results in Section~\ref{sec:numerical}.

\onehalfspacing

\section{Proofs} \label{sec:proof}

\begin{proof}[Proof of Lemma~\ref{lem:sigma-given-tau}]

\noindent
Integrating out $\bm\beta$ from the joint posterior $\pi(\bm\beta,\sigma^2,\bm\tau\mid\bm Y)$ yields
\bas
\pi(\sigma^2,\bm\tau\mid\bm Y)
&=\int\pi(\bm\beta,\sigma^2,\bm\tau\mid\bm Y)\;d\bm\beta\\
&=
\frac{(\lambda^2/2)^p}{(2\pi\sigma^2)^{(n+p+1)/2}}
\exp\left(-\frac{\lambda^2}2\sum_{j=1}^p\tau_j\right)\\
&\qquad\times\int\exp\left[-\frac1{2\sigma^2}\left(\left\|\tilde{\bm Y}-\bm X\bm\beta\right\|_2^2+\left\|\bm D_{\bm\tau}^{-1/2}\bm\beta\right\|_2^2\right)\right]\;d\bm\beta\\
&=
\frac{(\lambda^2/2)^p}{(2\pi\sigma^2)^{(n+p+1)/2}}
\exp\left[-\frac{\lambda^2\|\bm\tau\|_1}2
-\frac1{2\sigma^2}\tilde{\bm Y}^\T\left(\bm I_n-\bm X\bm A_{\bm\tau}^{-1}\bm X^\T\right)\tilde{\bm Y}
\right]\\
&\qquad\times\int\exp\left[-\frac1{2\sigma^2}\left\|\bm A_{\bm\tau}^{1/2}\left(\bm\beta-\bm A_{\bm\tau}^{-1}\bm X^\T\tilde{\bm Y}\right)\right\|_2^2\right]\;d\bm\beta\\
&=
\frac{(\lambda^2/2)^p}{(2\pi\sigma^2)^{(n+p+1)/2}\,\det(\bm A_{\bm\tau})^{1/2}}
\exp\left[-\frac{\lambda^2\|\bm\tau\|_1}2
-\frac{\tilde{\bm Y}^\T\left(\bm I_n-\bm X\bm A_{\bm\tau}^{-1}\bm X^\T\right)\tilde{\bm Y}}{2\sigma^2}
\right].
\eas
Then it is clear that
\bas
\pi(\sigma^2\mid\bm\tau,\bm Y)\propto\frac1{(\sigma^2)^{(n+p+1)/2}}\exp\left[-\frac{\tilde{\bm Y}^\T\left(\bm I_n-\bm X\bm A_{\bm\tau}^{-1}\bm X^\T\right)\tilde{\bm Y}}{2\sigma^2}\right],
\eas
and the result follows immediately.
\end{proof}

%

\begin{proof}[Proof of Lemma~\ref{lem:ratio}]

\noindent
Let $\bm X=\bm U\bm\Omega\bm V^\T$ be a singular value decomposition of~$\bm X$, where
$\omega_1,\ldots,\omega_{\min\{n,p\}}$ are the the singular values, and let $\tau_{\max}=\max_{1\le j\le p}\tau_j$.  Then
\bas
\frac{\left\|\bm A_{\bm\tau}^{-1}\bm X^\T\tilde{\bm Y}\right\|_2^2}{C_{\bm\tau}}
&=\frac{\tilde{\bm Y}^\T\bm U\bm\Omega\bm V^\T\left(\bm V\bm\Omega^\T\bm\Omega\bm V^\T+\bm D_{\bm\tau}^{-1}\right)^{-2}\bm V\bm\Omega^\T\bm U^\T\tilde{\bm Y}}{\tilde{\bm Y^\T}\tilde{\bm Y}-\tilde{\bm Y}^\T\bm U\bm\Omega\bm V^\T\left(\bm V\bm\Omega^\T\bm\Omega\bm V^\T+\bm D_{\bm\tau}^{-1}\right)^{
-1
}\bm V\bm\Omega^\T\bm U^\T\tilde{\bm Y}}\\
&\le\frac{\tilde{\bm Y}^\T\bm U\bm\Omega\bm V^\T\left(\bm V\bm\Omega^\T\bm\Omega\bm V^\T+\tau_{\max}^{-1}\bm I_p\right)^{-2}\bm V\bm\Omega^\T\bm U^\T\tilde{\bm Y}}{\tilde{\bm Y^\T}\tilde{\bm Y}-\tilde{\bm Y}^\T\bm U\bm\Omega\bm V^\T\left(\bm V\bm\Omega^\T\bm\Omega\bm V^\T+\tau_{\max}^{-1}\bm I_p\right)^{
-1
}\bm V\bm\Omega^\T\bm U^\T\tilde{\bm Y}}\\
&=\frac{\tilde{\bm Y}^\T\bm U\bm\Omega\left(\bm\Omega^\T\bm\Omega+\tau_{\max}^{-1}\bm I_p\right)^{-2}\bm\Omega^\T\bm U^\T\tilde{\bm Y}}{\tilde{\bm Y^\T}\bm U\bm U^\T\tilde{\bm Y}-\tilde{\bm Y}^\T\bm U\bm\Omega\left(\bm\Omega^\T\bm\Omega+\tau_{\max}^{-1}\bm I_p\right)^{
-1
}\bm\Omega^\T\bm U^\T\tilde{\bm Y}}
=\frac{\tilde{\bm Y_\star}^\T\bm G\tilde{\bm Y_\star}}{\tilde{\bm Y_\star}^\T\bm H\tilde{\bm Y_\star}},
\eas
where $\tilde{\bm Y}_\star=\bm U^\T\tilde{\bm Y}$ and
\bas
\bm G=\diag\left[\frac{\omega_1^2}{\left(\omega_1^2+\tau_{\max}^{-1}\right)^2},\ldots,\frac{\omega_n^2}{\left(\omega_n^2+\tau_{\max}^{-1}\right)^2}\right],\qquad
\bm H=\diag\left(\frac{\tau_{\max}^{-1}}{\omega_1^2+\tau_{\max}^{-1}},\ldots,\frac{\tau_{\max}^{-1}}{\omega_n^2+\tau_{\max}^{-1}}\right),
\eas
taking $\omega_i=0$ for all $i>p$ if $n>p$.  Then it is clear that
\bas
\frac{\left\|\bm A_{\bm\tau}^{-1}\bm X^\T\tilde{\bm Y}\right\|_2^2}{C_{\bm\tau}}
&\le\max_{1\le i\le n}\left[\frac{\omega_i^2}{(\omega_i^2+\tau_{\max}^{-1})^2}\left(\frac{\tau_{\max}^{-1}}{\omega_i^2+\tau_{\max}^{-1}}\right)^{-1}\right]\\
&\le\max_{1\le i\le n}\left[\frac1{4\tau_{\max}^{-1}}\left(\frac{\tau_{\max}^{-1}}{\omega_i^2+\tau_{\max}^{-1}}\right)^{-1}\right]\le\frac{\tau_{\max}}4\le\frac{\|\bm\tau\|_1}4,
\eas
noting for the second inequality that $a/(a+b)^2\le1/(4b)$ for all $a\ge0$ and $b>0$.
\end{proof}

\begin{proof}[Proof of Theorem \ref{thmsptlpty}]

\noindent
Let $\hat{K}$ be the Markov operator associated with the density $\hat{k}$. Let $\hat{K}^*$ denote the adjoint of 
$\hat{K}$. Note that $\hat{K}$ is Hilbert-Schmidt if and only if $\hat{K}^* \hat{K}$ is trace class, which happens if and only if 
$I < \infty$
\citep[see][]{jorgens1982}
where 
\begin{eqnarray}
I 
&:=& \int_{\mathbb{R}^p}\int_{\mathbb{R}_{+}} \int_{\mathbb{R}^p} \int_{\mathbb{R}_{+}}
\hat{k} \bigg(\left(\bm{\beta},\sigma^{2}\right), \left(\tilde{\bm{\beta}},\tilde{\sigma}^2\right)\bigg) 
\hat{k}^* \bigg(\left( \tilde{\bm{\beta}}, \tilde{\sigma}^{2} \right), \left( \bm{\beta}, \sigma^2 \right)\bigg) d\bm{\beta}d\sigma^2
d\tilde{\bm{\beta}}d\tilde{\sigma^2} \nonumber\\
&=& \int_{\mathbb{R}^p}\int_{\mathbb{R}_{+}} \int_{\mathbb{R}^p} \int_{\mathbb{R}_{+}}
\hat{k}^2\bigg(\left(\bm{\beta},\sigma^{2}\right), \left(\tilde{\bm{\beta}},\tilde{\sigma}^2\right)\bigg) 
\frac{ f\left(\bm{\beta},{\sigma^{2}} \mid {\bm Y}\right)  }{ f\left(\tilde{\bm{\beta}},\tilde{\sigma}^2 \mid {\bm Y}\right)  }d\bm{\beta} 
d\sigma^2d\tilde{\bm{\beta}}d\tilde{\sigma^2}. \label{eqtrace6}
\end{eqnarray}

\noindent
By (\ref{threebayes}), we get that 
\begin{eqnarray*}
& & \hat{k}^2\left(\left(\bm{\beta},\sigma^{2}\right), \left(\tilde{\bm{\beta}},\tilde{\sigma}^2\right)\right)\\ 
&=& \bigg[\int_{\mathbb{R}_{+}^p} 
f\left(\tilde{\sigma}^2 \mid \tilde{\bm{\beta}},\bm{\tau},{\bm Y}\right)
f\left(\tilde{\bm{\beta}} \mid \bm{\tau},{\sigma^{2}},{\bm Y}\right)
f\left(\bm{\tau} \mid {\bm{\beta}},\sigma^{2},{\bm Y}\right)
d\bm{\tau} \bigg]^2\\
&=&  \int_{\mathbb{R}_{+}^p} 
\int_{\mathbb{R}_{+}^p} 
f\left(\tilde{\sigma}^2 \mid \tilde{\bm{\beta}},\bm{\tau},{\bm Y}\right)
f\left(\tilde{\bm{\beta}} \mid \bm{\tau},{\sigma^{2}},{\bm Y}\right)
f\left(\bm{\tau} \mid {\bm{\beta}},\sigma^{2},{\bm Y}\right)\\
& & f\left(\tilde{\sigma}^2 \mid \tilde{\bm{\beta}},\tilde{\bm{\tau}},{\bm Y}\right)
f\left(\tilde{\bm{\beta}} \mid \tilde{\bm{\tau}},{\sigma^{2}},{\bm Y}\right)
f\left(\tilde{\bm{\tau}} \mid {\bm{\beta}},\sigma^{2},{\bm Y}\right)
d{\bm{{\tau}}} d\tilde{\bm{\tau}}. 
\end{eqnarray*}

\noindent
It follows from (\ref{eqtrace6}) and Fubini's theorem that 
\begin{eqnarray*}
I 
&=& \int_{\mathbb{R}_{+}^p} \int_{\mathbb{R}_{+}^p} \int_{\mathbb{R}^p} 
\int_{\mathbb{R}_{+}} \int_{\mathbb{R}^p} \int_{\mathbb{R}_{+}} 
f\left(\tilde{\sigma}^2~|~\tilde{\bm{\beta}},\bm{\tau},{\bm Y}\right)
f\left(\tilde{\bm{\beta}}~|~\bm{\tau},{\sigma^{2}},{\bm Y}\right)
f\left(\bm{\tau}~|~{\bm{\beta}},\sigma^{2},{\bm Y}\right)\\
& & f\left(\tilde{\sigma}^2~|~\tilde{\bm{\beta}},\tilde{\bm{\tau}},{\bm Y}\right)
f\left(\tilde{\bm{\beta}}~|~\tilde{\bm{\tau}},{\sigma^{2}},{\bm Y}\right)
f\left(\tilde{\bm{\tau}}~|~{\bm{\beta}},\sigma^{2},{\bm Y}\right)
\frac{ f\left(\bm{\beta},{\sigma^{2}}~|~{\bm Y}\right)  }{ f\left(\tilde{\bm{\beta}},\tilde{\sigma}^2~|~{\bm Y}\right)}\\
& & d{\bm{{\tau}}} d\tilde{\bm{\tau}} d\bm{\beta}d\sigma^2 d\tilde{\bm{\beta}} d\tilde{{\sigma^2}}.  
\end{eqnarray*}

\noindent
A straightforward manipulation of conditional densities shows that 
\begin{eqnarray*}
& & f\left(\tilde{\sigma}^2~|~\tilde{\bm{\beta}},\tilde{\bm{\tau}},{\bm Y}\right)
f\left(\tilde{\bm{\beta}}~|~\tilde{\bm{\tau}},{\sigma^{2}},{\bm Y}\right)
f\left(\tilde{\bm{\tau}}~|~{\bm{\beta}},\sigma^{2},{\bm Y}\right)
\frac{ f\left(\bm{\beta},{\sigma^{2}}~|~{\bm Y}\right)  }{ f\left(\tilde{\bm{\beta}},\tilde{\sigma}^2~|~{\bm Y}\right)}\\
&=& f\left(\bm{\beta}~|~\tilde{\bm{\tau}},{\sigma^{2}},{\bm Y}\right)
f\left({\sigma^{2}}~|~\tilde{\bm{\beta}},\tilde{\bm{\tau}},{\bm Y}\right)
f\left(\tilde{\bm{\tau}}~|~\tilde{\bm{\beta}},\tilde{\sigma}^2,{\bm Y}\right). 
\end{eqnarray*}

\noindent
It follows that 
\begin{eqnarray}
I 
&=& \int_{\mathbb{R}_{+}} \int_{\mathbb{R}_{+}} \int_{\mathbb{R}^p} 
\int_{\mathbb{R}^p} \int_{\mathbb{R}_{+}^p} \int_{\mathbb{R}_{+}^p} 
f\left(\tilde{\sigma}^2 \mid \tilde{\bm{\beta}},\bm{\tau},{\bm Y}\right)
f\left(\tilde{\bm{\beta}} \mid \bm{\tau},{\sigma^{2}},{\bm Y}\right)
f\left(\bm{\tau} \mid {\bm{\beta}},\sigma^{2},{\bm Y}\right) \nonumber\\
& & f\left(\bm{\beta} \mid \tilde{\bm{\tau}},{\sigma^{2}},{\bm Y}\right)
f\left({\sigma^{2}} \mid \tilde{\bm{\beta}},\tilde{\bm{\tau}},{\bm Y}\right)
f\left(\tilde{\bm{\tau}} \mid \tilde{\bm{\beta}},\tilde{\sigma}^2,{\bm Y}\right)
d\sigma^2 d\tilde{\sigma^2} d\bm{\beta} d\tilde{\bm{\beta}}
d{\bm{{\tau}}} d\tilde{\bm{\tau}} \label{eqtrace7}
\end{eqnarray}

\noindent
For convenience, we introduce and use the following notation in the subsequent proof. 
\begin{eqnarray}
& & \widehat{\bm{\beta}}={\bf A}_{\bm{\tau}}^{-1}X^T{\bm Y} \hspace{1.9in} \widehat{\bm{\beta}}_{*} = 
{\bf A}_{\tilde{\bm{\tau}}}^{-1}X^T{\bm Y} \nonumber\\
& & {\Delta}_{1}=(\tilde{\bm{\beta}}-\widehat{\bm{\beta}})^T{\bf A}_{\bm{\tau}} (\tilde{\bm{\beta}}-\widehat{\bm{\beta}}) 
\hspace{1.2in} {\Delta}_{1*}=(\bm{\beta}-\widehat{\bm{\beta}}_{*})^T{\bf A}_{\tilde{\bm{\tau}}}(\bm{\beta} - 
\widehat{\bm{\beta}}_{*}) \nonumber\\
& & \tilde{\Delta}= ({\bm Y}-X\tilde{\bm{\beta}})^T({\bm Y}-X\tilde{\bm{\beta}})+{\tilde{\bm{\beta}}}^TD_{\bm{\tau}}^{-1}
{\tilde{\bm{\beta}}}+2\xi \hspace{0.2in} \tilde{\Delta}_{*}= ({\bm Y}-X\tilde{\bm{\beta}})^T({\bm Y}-X\tilde{\bm{\beta}})+
{\tilde{\bm{\beta}}}^TD_{\tilde{\bm{\tau}}}^{-1}{\tilde{\bm{\beta}}}+2\xi. \nonumber\\
& & \label{eqtrace8}
\end{eqnarray}

\noindent
By (\ref{gibbs-general}) we get that 
\begin{eqnarray}
& & f\left(\tilde{\sigma}^2 \mid \tilde{\bm{\beta}},\bm{\tau},{\bm Y}\right)
f\left(\tilde{\bm{\beta}} \mid \bm{\tau},{\sigma^{2}},{\bm Y}\right)
f\left(\bm{\tau} \mid {\bm{\beta}},\sigma^{2},{\bm Y}\right)
f\left(\bm{\beta} \mid \tilde{\bm{\tau}},{\sigma^{2}},{\bm Y}\right) \times \nonumber\\
& & f\left({\sigma^{2}} \mid \tilde{\bm{\beta}},\tilde{\bm{\tau}},{\bm Y}\right)
f\left(\tilde{\bm{\tau}} \mid {\tilde{\bm{\beta}}},\tilde{\sigma}^2,{\bm Y}\right) \nonumber\\
&=& C_3\Bigg\{\frac{{\tilde{\Delta}}^{\frac{n+p+2\alpha}{2}} \exp \left(-\tilde{\Delta}/(2\tilde{\sigma}^2) \right)}
{{(\tilde{{\sigma}}^2)}^{{\frac{n+p+2\alpha}{2}}+1}} \Bigg\}\Bigg\{
\frac{{\left| {\bf A}_{\bm{\tau}} \right|}^{\frac{1}{2}} \exp \left(-{\Delta}_1/(2\sigma^2) \right)}{\sigma^p} \Bigg\} 
\times \nonumber\\
& & \Bigg\{\prod\limits_{j=1}^{p} (\tau_j)^{-\frac{1}{2}} \exp \left(-\frac{\beta_{j}^2}{2\sigma^2 
\tau_j}+\frac{\lambda}{\sigma} \left|\beta_j\right|-\frac{1}{2}\lambda^2\tau_j \right) \Bigg\} \Bigg\{
\frac{\left| {\bf A}_{\tilde{\bm{\tau}}} \right|^{\frac{1}{2}} \exp \left(-{\Delta}_{1*}/(2\sigma^2) \right)}{\sigma^p}
\Bigg\} \times \nonumber\\
& & \Bigg\{\frac{{\tilde{\Delta}}_{*}^{\frac{n+p+2\alpha}{2}} \exp \left(-\tilde{\Delta}_{*}/(2{\sigma^2}) \right)}
{{({\sigma^2})}^{{\frac{n+p+2\alpha}{2}}+1}}\Bigg\} \times \Bigg\{
\prod\limits_{j=1}^{p}(\tilde{\tau}_j)^{-\frac{1}{2}} \exp \left(-\frac{{\tilde{\beta_{j}}}^2}{2\tilde{\sigma}^2\tilde{\tau}_j} + 
\frac{\lambda}{\tilde{\sigma}}\left| \tilde{\beta_j} \right|-\frac{1}{2}\lambda^2\tilde{\tau}_j \right) \Bigg\} \nonumber\\
& & \nonumber\\
&\geq& C_3 f_1(\bm{\tau},\tilde{\bm{\tau}})\Bigg\{\frac{{\tilde{\Delta}}^{\frac{n+p+2\alpha}{2}} \exp \left(-\frac{{\tilde{\Delta}} + 
{\tilde{\bm{\beta}}}^TD_{\tilde{\bm{\tau}}}^{-1}\tilde{\bm{\beta}}}{2\tilde{{\sigma}}^2} \right)}{{(\tilde{\sigma}^2)}^{{\frac{n+p+2\alpha}{2}}+1}}
\Bigg\}\Bigg\{\frac{{\tilde{\Delta}}_{*}^{\frac{n+p+2\alpha}{2}} \exp \left(-\frac{{{\Delta}_{1}+\Delta_{1*}+\tilde{\Delta}_{*} + 
\bm{\beta}^TD_{\bm{\tau}}^{-1}\bm{\beta}}}{2{\sigma^2}} \right)}{{({\sigma^2})}^{{\frac{n+p+2\alpha}{2}}+p+1}}\Bigg\}, 
\label{eqtrace9}
\end{eqnarray}

\noindent
where 
$$
C_3 = \lambda^{2p} \left[(2\pi)^{{p}}2^{\frac{n+p+2\alpha}{2}}\Gamma(\frac{n+p+2\alpha}{2})\right]^{-2} 
$$

\noindent
and
$$
f_1(\bm{\tau},\tilde{\bm{\tau}})=\Bigg\{
\prod\limits_{j=1}^{p} (\tau_j)^{-\frac{1}{2}} \exp \left(-\frac{\lambda^2\tau_j}{2} \right)
\Bigg\}\Bigg\{
\prod\limits_{j=1}^{p}(\tilde{\tau}_j)^{-\frac{1}{2}} \exp \left(-\frac{\lambda^2\tilde{\tau}_j}{2} \right)
\Bigg\}{ \left| {\bf A}_{\bm{\tau}} \right|^{\frac 1 2 } \left| {\bf A}_{\tilde{\bm{\tau}}} \right|^{\frac 1 2}}. 
$$

\noindent
It follows by (\ref{eqtrace9}) and the form of the Inverse-Gamma density that 
\begin{eqnarray}
& & \int_{\mathbb{R}_{+}}\int_{\mathbb{R}_{+}} f\left(\tilde{\sigma}^2~|~\tilde{\bm{\beta}},\bm{\tau},{\bm Y}\right)
f\left(\tilde{\bm{\beta}}~|~\bm{\tau},{\sigma^{2}},{\bm Y}\right)
f\left(\bm{\tau}~|~{\bm{\beta}},\sigma^{2},{\bm Y}\right) \times \nonumber\\
& & f\left(\bm{\beta}~|~ \tilde{\bm{\tau}},{\sigma^{2}},{\bm Y}\right)
f\left({\sigma^{2}}~|~\tilde{\bm{\beta}},\tilde{\bm{\tau}},{\bm Y}\right)
f\left(\tilde{\bm{\tau}}~|~\tilde{\bm{\beta}},\tilde{\sigma}^2,{\bm Y}\right)d\sigma^2d\tilde{{\sigma^2}} \nonumber\\
&\geq& C_4f_1(\bm{\tau},\tilde{\bm{\tau}})\left\{\frac{{\tilde{\Delta}}^{\frac{n+p+2\alpha}{2}}}{\Big[{{\tilde{\Delta}} + 
{\tilde{\bm{\beta}}}^TD_{\tilde{\bm{\tau}}}^{-1}\tilde{\bm{\beta}}}\Big]^{{\frac{n+p+2\alpha}{2}}}}\right\}
\left\{\frac{{\tilde{\Delta}}_{*}^{\frac{n+p+2\alpha}{2}}}{\Big [{{{{\Delta}_{1}+\Delta_{1*}+\tilde{\Delta}_{*}+\bm{\beta}^T 
D_{\bm{\tau}}^{-1}\bm{\beta}}}}\Big ]^{\frac{n+p+2\alpha}{2}+p}}\right\}. \label{eqntrc10}
\end{eqnarray}

\noindent
where 
$$
C_4 = 2^{n+2p+2\alpha} \Gamma \left( \frac{n+p+2\alpha}{2} \right)^2 C_3. 
$$

\noindent
Note that 
\begin{eqnarray}
\Delta_{1*}+\bm{\beta}^TD_{\bm{\tau}}^{-1}\bm{\beta}
&=& (\bm{\beta}-\widehat{\bm{\beta}}_{*})^T{\bf A}_{\tilde{\bm{\tau}}}(\bm{\beta}-\widehat{\bm{\beta}}_{*}) + 
\bm{\beta}^TD_{\bm{\tau}}^{-1}\bm{\beta} \nonumber\\
&=& (\bm{\beta}-\widehat{\bm{\beta}}_{**})^T(X^TX+D_{\tilde{\bm{\tau}}}^{-1}+D_{\bm{\tau}}^{-1})(\bm{\beta} - 
\widehat{\bm{\beta}}_{**})+ \nonumber\\
& & (\widehat{\bm{\beta}}_{**}-\widehat{\bm{\beta}}_{*})^T{\bf A}_{\tilde{\bm{\tau}}}(\widehat{\bm{\beta}}_{**} - 
\widehat{\bm{\beta}}_{*}) + \widehat{\bm{\beta}}_{**}^TD_{\bm{\tau}}^{-1}\widehat{\bm{\beta}}_{**} \nonumber\\
&=& f_2(\bm{\tau},\tilde{\bm{\tau}})+(\bm{\beta}-\widehat{\bm{\beta}}_{**})^T(X^TX+D_{\tilde{\bm{\tau}}}^{-1} + 
D_{\bm{\tau}}^{-1})(\bm{\beta}-\widehat{\bm{\beta}}_{**}), \label{eqntrc11}
\end{eqnarray}

\noindent
where $\widehat{\bm{\beta}}_{**}= (X^TX+D_{\tilde{\bm{\tau}}}^{-1}+D_{\bm{\tau}}^{-1})^{-1}X^T{\bm Y}$ and 
$$
f_2(\bm{\tau},\tilde{\bm{\tau}})=(\widehat{\bm{\beta}}_{**}-\widehat{\bm{\beta}}_{*})^T{\bf A}_{\tilde{\bm{\tau}}} 
(\widehat{\bm{\beta}}_{**}-\widehat{\bm{\beta}}_{*}) + \widehat{\bm{\beta}}_{**}^TD_{\bm{\tau}}^{-1} 
\widehat{\bm{\beta}}_{**}. 
$$

\noindent
Hence, by (\ref{eqntrc11}) and the form of the multivariate $t$-distribution
\citep[see][for example]{kotz2004},
we obtain that 
\begin{eqnarray}
& & \int_{\mathbb{R}^p}
\frac{1}{\Big [{{{{\Delta}_{1}+\Delta_{1*}+\tilde{\Delta}_{*}+\bm{\beta}^TD_{\bm{\tau}}^{-1}\bm{\beta}}}}\Big ]^{\frac{n+p+2\alpha}
{2}+p}}d\bm{\beta} \nonumber\\
&=&  \int_{\mathbb{R}^p}  \frac{1}{\left[{\Delta}_{1}+\tilde{\Delta}_{*}+f_2(\bm{\tau},\tilde{\bm{\tau}})+(\bm{\beta} - 
\widehat{\bm{\beta}}_{**})^T(X^TX+D_{\tilde{\bm{\tau}}}^{-1}+D_{\bm{\tau}}^{-1})(\bm{\beta}-\widehat{\bm{\beta}}_{**}) 
\right]^{\frac{p+\left(n+2p+2\alpha\right)}{2}} } d\bm{\beta} \nonumber\\
&=& \frac{\Gamma\left(\frac{n+2p+2\alpha}{2}\right)\sqrt{\pi}^p\left|(X^TX+D_{\tilde{\bm{\tau}}}^{-1}+D_{\bm{\tau}}^{-1})\right|^{-\frac{1}
{2}}}{\Gamma\left(\frac{n+p+2\alpha}{2}+p\right) \left[{\Delta}_{1}+\tilde{\Delta}_{*}+f_2(\bm{\tau},\tilde{\bm{\tau}})\right]^{\frac{n+p
+2\alpha }{2}+\frac{p}{2}}}. \label{eqntrc12}
\end{eqnarray}

\noindent
It follows from (\ref{eqntrc10}) and (\ref{eqntrc12}) that 
\begin{eqnarray}
& & \int_{\mathbb{R}_{+}}\int_{\mathbb{R}_{+}}\int_{\mathbb{R}^{p}} f\left(\tilde{\sigma}^2~|~\tilde{\bm{\beta}},\bm{\tau},{\bm Y}\right)
f\left(\tilde{\bm{\beta}} \mid \bm{\tau},{\sigma^{2}},{\bm Y}\right)
f\left(\bm{\tau} \mid {\bm{\beta}},\sigma^{2},{\bm Y}\right) \times \nonumber\\
& & f\left(\bm{\beta} \mid \tilde{\bm{\tau}},{\sigma^{2}},{\bm Y}\right)
f\left({\sigma^{2}} \mid \tilde{\bm{\beta}},\tilde{\bm{\tau}},{\bm Y}\right)
f\left(\tilde{\bm{\tau}} \mid \tilde{\bm{\beta}},\tilde{\sigma}^2,{\bm Y}\right)d\sigma^2 d\tilde{\sigma}^2 d\bm{\beta} \nonumber\\
&\geq& C_4 f_1(\bm{\tau},\tilde{\bm{\tau}})\left\{\frac{{\tilde{\Delta}}^{\frac{n+p+2\alpha}{2}}{\tilde{\Delta}_{*}}^{\frac{n + p + 2\alpha}
{2}}}{\Big[{{\tilde{\Delta}}+{\tilde{\bm{\beta}}}^TD_{\tilde{\bm{\tau}}}^{-1}\tilde{\bm{\beta}}}\Big]^{{\frac{n+p+2\alpha}{2}}}}\right\}
\left\{\frac{\Gamma\left(\frac{n+2p+2\alpha}{2}\right)\sqrt{\pi}^p\left|(X^TX+D_{\tilde{\bm{\tau}}}^{-1}+D_{\bm{\tau}}^{-1})\right|^{-
\frac{1}{2}}}{\Gamma\left(\frac{n+p+2\alpha}{2}+p\right) \left[{\Delta}_{1}+\tilde{\Delta}_{*}+f_2(\bm{\tau},\tilde{\bm{\tau}})
\right]^{\frac{n+p+2\alpha }{2}+\frac{p}{2}}} \right\}. \nonumber\\
& & \label{eqntrc13}
\end{eqnarray}

\noindent
Note that 
\begin{equation} \label{eqntrc14}
\frac{{\tilde{\bm{\beta}}}^TD_{\tilde{\bm{\tau}}}^{-1}\tilde{\bm{\beta}}}{\tilde{\Delta}}=\frac{{\tilde{\bm{\beta}}}^T 
D_{\tilde{\bm{\tau}}}^{-1}\tilde{\bm{\beta}}}{({\bm Y}-X\tilde{\bm{\beta}})^T({\bm Y}-X\tilde{\bm{\beta}}) + 
{\tilde{\bm{\beta}}}^TD_{\bm{\tau}}^{-1}{\tilde{\bm{\beta}}}+2\xi  }\leq \frac{{\tilde{\bm{\beta}}}^T D_{\tilde{\bm{\tau}}}^{-1} 
\tilde{\bm{\beta}}}{{\tilde{\bm{\beta}}}^TD_{\bm{\tau}}^{-1}\tilde{\bm{\beta}}}\leq \max\limits_{1\leq j\leq p} 
\left( \frac{\tau_j}{\tilde{\tau}_j}\right), 
\end{equation}

\noindent
and 
\begin{eqnarray}
\frac{\Delta_1}{\tilde{\Delta}_{*}}
&=& \frac{(\tilde{\bm{\beta}}-\widehat{\bm{\beta}})^T({\bf A}_{\bm{\tau}})(\tilde{\bm{\beta}}-\widehat{\bm{\beta}})}{ ({\bm Y}-X\tilde{\bm{\beta}})^T({\bm Y}-X\tilde{\bm{\beta}})+{\tilde{\bm{\beta}}}^TD_{\tilde{\bm{\tau}}}^{-1}{\tilde{\bm{\beta}}}+2\xi} 
\nonumber\\
&\leq& \frac{(\tilde{\bm{\beta}}-\widehat{\bm{\beta}})^T({\bf A}_{\bm{\tau}})(\tilde{\bm{\beta}}-\widehat{\bm{\beta}})+ 
\widehat{\bm{\beta}}^TD_{\bm{\tau}}^{-1}\widehat{\bm{\beta}}+\left({\bm Y}-X\widehat{\bm{\beta}}\right)^T\left({\bm Y}-X
\widehat{\bm{\beta}}\right)}{ ({\bm Y}-X\tilde{\bm{\beta}})^T({\bm Y}-X\tilde{\bm{\beta}})+{\tilde{\bm{\beta}}}^T 
D_{\tilde{\bm{\tau}}}^{-1}{\tilde{\bm{\beta}}}+2\xi} \nonumber\\
&=& \frac{\left({\bm Y}-X\tilde{\bm{\beta}}\right)^T\left({\bm Y}-X\tilde{\bm{\beta}}\right)+{\tilde{\bm{\beta}}}^TD_{\bm{\tau}}^{-1} 
\tilde{\bm{\beta}}}{{ ({\bm Y}-X\tilde{\bm{\beta}})^T({\bm Y}-X\tilde{\bm{\beta}})+{\tilde{\bm{\beta}}}^TD_{\tilde{\bm{\tau}}}^{-1}
{\tilde{\bm{\beta}}}+2\xi}} \nonumber\\
&\leq& 1+ \frac{{\tilde{\bm{\beta}}}^TD_{\bm{\tau}}^{-1}\tilde{\bm{\beta}}}{{\tilde{\bm{\beta}}}^TD_{\tilde{\bm{\tau}}}^{-1} 
\tilde{\bm{\beta}}} \nonumber\\
&\leq& 1 + \max\limits_{1\leq j\leq p}\left(\frac{\tilde{\tau}_j}{{\tau_j}}\right). \label{eqntrc15}
\end{eqnarray}

\noindent
From (\ref{eqntrc13}), (\ref{eqntrc14}), (\ref{eqntrc15}), and the fact 
$$
\tilde{\Delta}_{*}= {{ ({\bm Y}-X\tilde{\bm{\beta}})^T({\bm Y}-X\tilde{\bm{\beta}})+{\tilde{\bm{\beta}}}^TD_{\tilde{\bm{\tau}}}^{-1}
{\tilde{\bm{\beta}}}+2\xi}}\geq {{{\widehat{\bm{\beta}}}_{*}^TD_{\tilde{\bm{\tau}}}^{-1}{{\widehat{\bm{\beta}}}_{*}}+ ({\bm Y}-
X{{\widehat{\bm{\beta}}}_{*}})^T({\bm Y}-X{{\widehat{\bm{\beta}}}_{*}})+2\xi}} 
$$

\noindent
($\widehat{\bm{\beta}}_{*}$ minimizes the L.H.S. as a function of $\tilde{\bm{\beta}}$), it follows that 
\begin{eqnarray}
& & \int_{\mathbb{R}^p}\int_{\mathbb{R}_{+}}\int_{\mathbb{R}_{+}} f\left(\tilde{\sigma}^2 \mid \tilde{\bm{\beta}},\bm{\tau},{\bm Y}
\right) f\left(\tilde{\bm{\beta}} \mid \bm{\tau},{\sigma^{2}},{\bm Y}\right) f\left(\bm{\tau} \mid {\bm{\beta}},\sigma^{2},{\bm Y}\right)
\nonumber\\
& & f\left(\bm{\beta} \mid \tilde{\bm{\tau}},{\sigma^{2}},{\bm Y}\right)
f\left({\sigma^{2}} \mid \tilde{\bm{\beta}},\tilde{\bm{\tau}},{\bm Y}\right)
f\left(\tilde{\bm{\tau}} \mid \tilde{\bm{\beta}},\tilde{\sigma}^2,{\bm Y}\right)d\sigma^2 d\tilde{\sigma}^2 d\bm{\beta} \nonumber\\
&\geq& f_3(\bm{\tau},\tilde{\bm{\tau}}) \left\{ \frac{1}{\left[{\Delta}_{1}+\tilde{\Delta}_{*}+f_2(\bm{\tau},\tilde{\bm{\tau}}) 
\right]^{\frac{p}{2}}}\right\}, \label{eqntrc16}
\end{eqnarray}

\noindent
where 
\begin{eqnarray*}
f_3(\bm{\tau},\tilde{\bm{\tau}}) 
&=& C_4 f_1(\bm{\tau},\tilde{\bm{\tau}})
\left\{\frac{\Gamma\left(\frac{n+2p+2\alpha}{2}\right)\sqrt{\pi}^p\left|X^TX+D_{\tilde{\bm{\tau}}}^{-1}+D_{\bm{\tau}}^{-1} 
\right|^{-\frac{1}{2}}}{\Gamma\left(\frac{n+p+2\alpha}{2}+p\right) \left[1+\max\limits_{1\leq j\leq p}\left(\frac{\tau_j}{\tilde{\tau}_j} 
\right)\right]^{\frac{n+p+2\alpha }{2}}} \right\} \times\\
& & \left[ 2+\max\limits_{1\leq j\leq p}\left(\frac{\tilde{\tau}_j}{{\tau_j}}\right)+\frac{f_2(\bm{\tau},\tilde{\bm{\tau}})}
{{{{\widehat{\bm{\beta}}}_{*}^TD_{\tilde{\bm{\tau}}}^{-1}{{\widehat{\bm{\beta}}}_{*}}+ ({\bm Y}-X{{\widehat{\bm{\beta}}}_{*}})^T 
({\bm Y}-X{{\widehat{\bm{\beta}}}_{*}})+2\xi}}} \right]^{-\frac{n+p+2\alpha}{2}}. 
\end{eqnarray*}

\noindent
Now, note that
\begin{eqnarray*}
& & {{\Delta}_{1}+\tilde{\Delta}_{*}+f_2(\bm{\tau},\tilde{\bm{\tau}})}\\
&=& {(\tilde{\bm{\beta}}-\widehat{\bm{\beta}})^T({\bf A}_{\bm{\tau}})(\tilde{\bm{\beta}}-\widehat{\bm{\beta}})}+{ ({\bm Y}-X
\tilde{\bm{\beta}})^T({\bm Y}-X\tilde{\bm{\beta}})+{\tilde{\bm{\beta}}}^TD_{\tilde{\bm{\tau}}}^{-1}{\tilde{\bm{\beta}}}+2\xi} + 
f_2(\bm{\tau},\tilde{\bm{\tau}})\\
&=& {(\tilde{\bm{\beta}}-\widehat{\bm{\beta}})^T({\bf A}_{\bm{\tau}})(\tilde{\bm{\beta}}-\widehat{\bm{\beta}})}
+{(\tilde{\bm{\beta}}-\widehat{\bm{\beta}}_{*})^T{\bf A}_{\tilde{\bm{\tau}}}(\tilde{\bm{\beta}}-\widehat{\bm{\beta}}_{*})} +\\
& & {({\bm Y}-X{\widehat{\bm{\beta}}_{*}})^T({\bm Y}-X{\widehat{\bm{\beta}}_{*}})+{{\widehat{\bm{\beta}}_{*}}}^T 
D_{\tilde{\bm{\tau}}}^{-1}{{\widehat{\bm{\beta}}_{*}}}}+2\xi+f_2(\bm{\tau},\tilde{\bm{\tau}})\\
&=& {(\tilde{\bm{\beta}}-\widehat{\bm{\beta}}_{***})^T({\bf A}_{\bm{\tau}}+{\bf A}_{\tilde{\bm{\tau}}})(\tilde{\bm{\beta}}-
\widehat{\bm{\beta}}_{***})} +\\
& & {({\widehat{\bm{\beta}}_{***}}-\widehat{\bm{\beta}})^T({\bf A}_{\bm{\tau}})({\widehat{\bm{\beta}}_{***}}-\widehat{\bm{\beta}})}
+ {({\widehat{\bm{\beta}}_{***}}-\widehat{\bm{\beta}}_{*})^T{\bf A}_{\tilde{\bm{\tau}}}({\widehat{\bm{\beta}}_{***}}-
\widehat{\bm{\beta}}_{*})} +\\
& & {({\bm Y}-X{\widehat{\bm{\beta}}_{*}})^T({\bm Y}-X{\widehat{\bm{\beta}}_{*}})+{{\widehat{\bm{\beta}}_{*}}}^T 
D_{\tilde{\bm{\tau}}}^{-1}{{\widehat{\bm{\beta}}_{*}}}}+2\xi+f_2(\bm{\tau},\tilde{\bm{\tau}})\\
&=& {(\tilde{\bm{\beta}}-\widehat{\bm{\beta}}_{***})^T({\bf A}_{\bm{\tau}}+{\bf A}_{\tilde{\bm{\tau}}})(\tilde{\bm{\beta}}-
\widehat{\bm{\beta}}_{***})}+f_4(\bm{\tau},\tilde{\bm{\tau}}), 
\end{eqnarray*}

\noindent
where $\hat{\bm{\beta}}_{***}=({\bf A}_{\bm{\tau}}+{\bf A}_{\tilde{\bm{\tau}}})^{-1} 2 X^T{\bm Y}$, and
\begin{eqnarray*}
f_4(\bm{\tau},\tilde{\bm{\tau}}) 
&=& {({\widehat{\bm{\beta}}_{***}}-\widehat{\bm{\beta}})^T({\bf A}_{\bm{\tau}})({\widehat{\bm{\beta}}_{***}}-\widehat{\bm{\beta}})} 
+ {({\widehat{\bm{\beta}}_{***}}-\widehat{\bm{\beta}}_{*})^T{\bf A}_{\tilde{\bm{\tau}}}({\widehat{\bm{\beta}}_{***}} - 
\widehat{\bm{\beta}}_{*})} +\\
& & {({\bm Y}-X{\widehat{\bm{\beta}}_{*}})^T({\bm Y}-X{\widehat{\bm{\beta}}_{*}})+{{\widehat{\bm{\beta}}_{*}}}^T 
D_{\tilde{\bm{\tau}}}^{-1}{{\widehat{\bm{\beta}}_{*}}}}+2\xi+f_2(\bm{\tau},\tilde{\bm{\tau}}). 
\end{eqnarray*}

\noindent
It follows from (\ref{eqntrc16}) that 
\begin{eqnarray*}
& & \int_{\mathbb{R}_{+}}\int_{\mathbb{R}_{+}}\int_{\mathbb{R}^p}\int_{\mathbb{R}^p} f\left(\tilde{\sigma}^2 \mid 
\tilde{\bm{\beta}},\bm{\tau},{\bm Y}\right) f\left(\tilde{\bm{\beta}} \mid \bm{\tau},{\sigma^{2}},{\bm Y}\right) 
f\left(\bm{\tau} \mid {\bm{\beta}},\sigma^{2},{\bm Y}\right)\\
& & f\left(\bm{\beta} \mid \tilde{\bm{\tau}},{\sigma^{2}},{\bm Y}\right) f\left({\sigma^{2}} \mid \tilde{\bm{\beta}},\tilde{\bm{\tau}}, 
{\bm Y}\right) f\left(\tilde{\bm{\tau}} \mid \tilde{\bm{\beta}},\tilde{\sigma}^2,{\bm Y}\right)d\sigma^2d \tilde{\sigma}^2d\bm{\beta} 
d\tilde{\bm{\beta}}\\
&\geq& f_3(\bm{\tau},\tilde{\bm{\tau}})\int_{\mathbb{R}^p}\left\{\frac{1}{\left[f_4(\bm{\tau},\tilde{\bm{\tau}})+{(\tilde{\bm{\beta}}-
\widehat{\bm{\beta}}_{***})^T~({\bf A}_{\bm{\tau}}+{\bf A}_{\tilde{\bm{\tau}}})~(\tilde{\bm{\beta}}-\widehat{\bm{\beta}}_{***})}
\right]^{\frac{p}{2}}}d\tilde{\bm{\beta}}\right\}\\
&\geq& \frac{f_3(\bm{\tau},\tilde{\bm{\tau}})}{\left[f_4(\bm{\tau},\tilde{\bm{\tau}})\right]^{\frac{p}{2}}} \int_{\mathbb{R}^p} 
\left\{ \frac{1}{\left[1+{(\tilde{\bm{\beta}}-\widehat{\bm{\beta}}_{***})^T~\frac{({\bf A}_{\bm{\tau}}+{\bf A}_{\tilde{\bm{\tau}}})}
{f_4(\bm{\tau},\tilde{\bm{\tau}})}~(\tilde{\bm{\beta}}-\widehat{\bm{\beta}}_{***})}\right]^{\frac{p}{2}}}d\tilde{\bm{\beta}} \right\}\\
&=& \infty 
\end{eqnarray*}

\noindent
for every $\left(\bm{\tau},\tilde{\bm{\tau}}\right) \in {\mathbb{R}_{+}^{p}} \times {\mathbb{R}_{+}^{p}}. $ The last integral is infinite 
based on the fact that the multivariate $t$-distribution with $1$ degree of freedom has an infinite mean
\citep[see][]{kotz2004}.
The first part of the result now follows from (\ref{eqtrace7}). 

We now focus on proving the second part of the result. In the current setting, to prove the trace class property for $k$, we need 
to show
\citep[see][]{jorgens1982}
that 
\begin{equation} \label{eqtrace1}
\iint\limits_{\mathbb{R}^{p} \times \mathbb{R}_{+}} k \left(\left(\bm{\beta},\sigma^{2}\right), 
\left(\bm{\beta},\sigma^{2}\right)\right)\,d{\bm{\beta}}\,d{\sigma^2}< \infty. 
\end{equation}

\noindent
It follows by (\ref{eqtwobayes}) that 
\begin{eqnarray*}
& & \iint\limits_{\mathbb{R}^{p} \times \mathbb{R}_{+}}  k \left(\left(\bm{\beta},\sigma^{2}\right), \left(\bm{\beta},\sigma^{2}\right)\right)\,d{\bm{\beta}}\,d{\sigma^2}\\
&=& \iiint\limits_{\mathbb{R}^p \times \mathbb{R}_{+}^{p}\times \mathbb{R}_{+}}
f\left(\bm{\tau}~|~{\bm{\beta}},\sigma^{2},{\bm Y}\right)f\left(\sigma^{2}~|~\bm{\tau},{\bm Y}\right)f\left(\bm{\beta}~|~\sigma^{2},\bm{\tau},{\bm Y}\right)\,d{\bm{\beta}}\,d{\bm{\tau}}\,d{\sigma^2}\\
&=& C_1\iiint\limits_{\mathbb{R}^p \times \mathbb{R}_{+}^{p}\times \mathbb{R}_{+}}
\left(\prod\limits_{j=1}^p\sqrt{\frac{\lambda^2}{2\pi}}\left(\tau_j\right)^{-\frac{1}{2}} \exp \left(-\frac{\beta_j^2\left(1-\tau_j\sqrt{\frac{\lambda^2\sigma^2}{\beta_j^2}}\right)^2}{2\sigma^2\tau_j} \right) \right)
\left(\sigma^2\right)^{-\frac{n+2\alpha}{2}-1}\\
& & \exp \left( -\frac{1}{2\sigma^2}\left({{\bm Y}}^T\left(I-X{\bf A}_{\bm{\tau}}^{-1}X^T\right){{\bm Y}}+2\xi\right) \right) 
\left({{\bm Y}}^T\left(I-X{\bf A}_{\bm{\tau}}^{-1}X^T\right){{\bm Y}}+2\xi\right)^{\frac{n+2\alpha}{2}}(\sigma^2)^{-\frac{p}{2}}\\ 
& & \times |{\bf A}_{\bm{\tau}}~|~^{\frac{1}{2}} {\times} \exp \left( -\frac{\left(\bm{\beta}-{\bf A}_{\bm{\tau}}^{-1}X^T 
{\bm Y}\right)^T{\bf A}_{\bm{\tau}}\left(\bm{\beta}-{\bf A}_{\bm{\tau}}^{-1}X^T{\bm Y}\right)}{2\sigma^2} \right) \,d{\bm{\beta}}
\,d{\bm{\tau}}\,d{\sigma^2}\\
&\stackrel{(a)}{\leq}& C_2 \iiint\limits_{\mathbb{R}^p \times \mathbb{R}_{+}^{p}\times \mathbb{R}_{+}}
\frac{\exp \left(-\frac{\xi}{\sigma^2} \right)}{\left(\sigma^2\right)^{\frac{n+2\alpha}{2}+1}}\frac{\exp \left(-\frac{1}{2\sigma^2}\sum_{j=1}^{p}\frac{\beta_j^2}{\tau_j}+\frac{\lambda}{\sigma}\sum_{j=1}^{p}~|\beta_j|~-\frac{\lambda^2}{2}\sum_{j=1}^{p}\tau_j \right)}{\prod_{j=1}^{p}\left(\tau_j\right)^{\frac{1}{2}}}\frac{~|~{\bf A}_{\bm{\tau}}~|~^\frac{1}{2}}{\left(\sigma^2\right)^\frac{p}{2}}\\
& & \times~\exp \left(-\frac{1}{2\sigma^2}\left(\bm{\beta}^T{\bf A}_{\bm{\tau}}\bm{\beta}-2\bm{\beta}^TX^T{\bm Y}+{{\bm Y}}^T{\bm Y}\right) \right)
\,d{\bm{\beta}}\,d{\bm{\tau}}\,d{\sigma^2}, 
\end{eqnarray*}

\noindent
where 
$$
C_1=\frac{1}{(2\pi)^{\frac p 2}2^{\frac {n+2\alpha} 2}\Gamma\left(\frac{n+2\alpha}{2}\right)} \hspace{0.2in} \mbox{ and } \hspace{0.2in} C_2 =  {\left({{\bm Y}}^T{{\bm Y}}
+2\xi\right)^{\frac{n+2\alpha}{2}}}\left(\frac{\lambda^2}{2\pi}\right)^{\frac{p}{2}}C_1. 
$$

\noindent
Also, $(a)$ follows from the fact that ${{\bm Y}}^T\left(I-X{\bf A}_{\bm{\tau}}^{-1}X^T\right){{\bm Y}} \leq {{\bm Y}^T}{{\bm Y}}$. For 
convenience, let 
$$
h\left(\sigma^2\right)= \exp \left(-\frac{\xi}{\sigma^2} \right) \left(\sigma^2\right)^{-\frac{n+2\alpha}{2}-1}. 
$$

\noindent 
It follows that 
\begin{eqnarray*}
& & \iint\limits_{\mathbb{R}^{p} \times \mathbb{R}_{+}} k \left(\left(\bm{\beta},\sigma^{2}\right), \left(\bm{\beta},\sigma^{2} 
\right)\right)\,d{\bm{\beta}}\,d{\sigma^2}\\
&\leq& C_2 \iiint\limits_{\mathbb{R}^p \times \mathbb{R}_{+}^{p}\times \mathbb{R}_{+}}
h\left(\sigma^2\right)\frac{\exp \left(-\frac{1}{2\sigma^2}\sum_{j=1}^{p}\frac{\beta_j^2}{\tau_j}+\frac{\lambda}{\sigma}\sum_{j=1}^{p}~|\beta_j|~-\frac{\lambda^2}{2}\sum_{j=1}^{p}\tau_j \right)}{\prod_{j=1}^{p}\left(\tau_j\right)^{\frac{1}{2}}}\frac{~|~{\bf A}_{\bm{\tau}}~|~^\frac{1}{2}}{\left(\sigma^2\right)^\frac{p}{2}}\\
& & \times~\exp \left(-\frac{1}{2\sigma^2}\left(\bm{\beta}^T{\bf A}_{\bm{\tau}}\bm{\beta}-2\bm{\beta}^TX^T{\bm Y}+{{\bm Y}}^T{\bm Y}\right) \right)
\,d{\bm{\beta}}\,d{\bm{\tau}}\,d{\sigma^2}\\
&=& C_2
\iiint\limits_{\mathbb{R}^p \times \mathbb{R}_{+}^{p}\times \mathbb{R}_{+}}
\frac{h\left(\sigma^2\right)}{\prod_{j=1}^{p}\left(\tau_j\right)^{\frac{1}{2}}}\exp \left(-\frac{1}{2\sigma^2}\sum_{j=1}^{p}\frac{\beta_j^2}{\tau_j}+\frac{\lambda}{\sigma}\sum_{j=1}^{p}~|\beta_j|~-\frac{\lambda^2}{2}\sum_{j=1}^{p}\tau_j \right)
\frac{~|~{\bf A}_{\bm{\tau}}~|~^\frac{1}{2}}{\left(\sigma^2\right)^\frac{p}{2}}\\
& & \times~\exp \left(-\frac{\bm{\beta}^TD_{\bm{\tau}}^{-1}\bm{\beta}}{4\sigma^2} \right) \exp \left(-\frac{\bm{\beta}^T 
D_{\bm{\tau}}^{-1}\bm{\beta}}{4\sigma^2} \right) \exp \left(-\frac{\left({{\bm Y}}-X\bm{\beta}\right)^T\left({{\bm Y}}-X\bm{\beta}\right)}{2\sigma^2} \right)\,d{\bm{\beta}}\,d{\bm{\tau}}\,d{\sigma^2}\\
&\leq& C_2
\iiint\limits_{\mathbb{R}^p \times \mathbb{R}_{+}^{p}\times \mathbb{R}_{+}}
\frac{h\left(\sigma^2\right)}{\prod_{j=1}^{p}\left(\tau_j\right)^{\frac{1}{2}}} \exp \left(-\left(\frac{1}{2}+\frac{1}{4}\right){\frac{1}{\sigma^2}\sum_{j=1}^{p}\frac{\beta_j^2}{\tau_j}}+\frac{\lambda}{\sigma}\sum_{j=1}^{p}~|\beta_j|~-{\lambda^2} \left( \frac{1}{3} + \frac{1}{6} \right) 
\sum_{j=1}^{p}\tau_j \right)\\
& & \times \frac{~|~{\bf A}_{\bm{\tau}}~|~^\frac{1}{2}}{\left(\sigma^2\right)^\frac{p}{2}} \exp \left(-\frac{\bm{\beta}^T 
D_{\bm{\tau}}^{-1}\bm{\beta}}{4\sigma^2} \right) \,d{\bm{\beta}}\,d{\bm{\tau}}\,d{\sigma^2}\\
&=& C_2
\iiint\limits_{\mathbb{R}_{+} \times \mathbb{R}_{+}^{p} \times \mathbb{R}^p}
\frac{h\left(\sigma^2\right)}{\prod_{j=1}^{p}\left(\tau_j\right)^{\frac{1}{2}}} \exp \left(-\left({{\frac{3}{4\sigma^2}\sum_{j=1}^{p}\frac{\beta_j^2}{\tau_j}}}-\frac{\lambda}{\sigma}\sum_{j=1}^{p}~|\beta_j|~+\frac{\lambda^2}{3}\sum_{j=1}^{p}\tau_j\right) 
\right) \frac{~|~{\bf A}_{\bm{\tau}}~|~^\frac{1}{2}}{\left(\sigma^2\right)^\frac{p}{2}}\\
& & \times \exp \left( -\frac{\lambda^2}{6}\sum_{j=1}^{p}\tau_j \right) \exp \left(-\frac{\bm{\beta}^TD_{\bm{\tau}}^{-1}\bm{\beta}}
{4\sigma^2} \right) \,d{\bm{\beta}} \,d{\bm{\tau}}\,d{\sigma^2}. 
\end{eqnarray*}

\noindent
Note that 
$$
\exp \left(-\left({{\frac{3}{4\sigma^2}\sum_{j=1}^{p}\frac{\beta_j^2}{\tau_j}}}-\frac{\lambda}{\sigma}\sum_{j=1}^{p}~|\beta_j|~ +
\frac{\lambda^2}{3}\sum_{j=1}^{p}\tau_j\right) \right) = \prod_{j=1}^{p}\exp \left(-\left(\frac{\sqrt{3\bm{\beta}_j^2}}{2\sigma\sqrt{{\tau_j}}} - 
\frac{\sqrt{\lambda^2\tau_j}}{\sqrt{3}}\right)^2 \right) \leq 1, 
$$

\noindent
and 
$$
\int_{\mathbb{R}^p} \exp \left(-\frac{\bm{\beta}^TD_{\bm{\tau}}^{-1}\bm{\beta}}{4\sigma^2} \right) \,d{\bm{\beta}} = 
\frac{\left(4\pi\right)^{p/2}\left( \sigma^2 \right)^{p/2}}{\left| D_{\bm{\tau}}^{-1} \right|^{1/2}}. 
$$

\noindent
It follows that 
\begin{eqnarray*}
& & \iint\limits_{\mathbb{R}^p \times \mathbb{R}_{+}} k \left( \left(\bm{\beta},\sigma^{2}\right), \left(\bm{\beta},\sigma^{2}\right) 
\right) \, d{\bm{\beta}} \, d{\sigma^2}\\
&\leq& \left(4\pi\right)^{\frac p 2}C_2 \iint\limits_{\mathbb{R}_{+} \times \mathbb{R}_{+}^{p}}
\frac{h\left(\sigma^2\right)}{\prod_{j=1}^{p}\left(\tau_j\right)^{\frac{1}{2}}} \exp \left(-\frac{\lambda^2}{6}\sum_{j=1}^{p}\tau_j 
\right) \frac{~|~{\bf A}_{\bm{\tau}}~|~^\frac{1}{2}}{\left(\sigma^2\right)^\frac{p}{2}}\frac{\left(\sigma^2\right)^\frac{p}{2}}{~|~D_{\bm{\tau}}^{-1}~|~^\frac{1}{2}}\,d{\bm{\tau}}\,d{\sigma^2} \\
&\leq& \left(4\pi\right)^{\frac p 2}C_2 \left(\int_{\mathbb{R}_{+}}h\left(\sigma^2\right)\,d{\sigma^2}\right)
\left(\int_{\mathbb{R}_{+}^{p}}\frac{~|~{\bf A}_{\bm{\tau}}~|~^\frac{1}{2}}{~|~D_{\bm{\tau}}^{-1}~|~^\frac{1}{2}}\frac{1}{\prod_{j=1}^p \left(\tau_j\right)^{\frac{1}{2}}} \exp \left(-\frac{\lambda^2}{6}\sum_{j=1}^{p}\tau_j \right) \,d{\bm{\tau}} \right)\\
&=& \left(4\pi\right)^{\frac p 2}C_2 ~\Gamma\left(\frac{n+2\alpha}{2}\right) \xi^{-\frac{n+2\alpha}{2}} 
\int_{\mathbb{R}_{+}^{p}} |{\bf A}_{\bm{\tau}}|^\frac{1}{2} \exp \left(-\frac{\lambda^2}{6}\sum_{j=1}^{p}\tau_j \right)\,d{\bm{\tau}}
\end{eqnarray*}

\noindent
Hence, to prove the required result, it is enough to show that 
\begin{equation} \label{eqtrace2}
\int_{\mathbb{R}_{+}^{p}} |~{\bf A}_{\bm{\tau}}|~^\frac{1}{2} \exp \left(-\frac{\lambda^2}{6}\sum_{j=1}^{p}\tau_j \right) \,d{\bm{\tau}} 
< \infty. 
\end{equation}

\noindent
Let $c^2$ denote the maximum eigenvalue of $X^T X$. Then, by
the results of \citet{fiedler1971},
it follows that 
\begin{eqnarray}
|{\bf A}_{\bm{\tau}}|^{1/2} 
&\leq& \prod_{j=1}^{p} \sqrt{\left(c^2+\frac{1}{\tau_j}\right)} \leq \prod_{j=1}^p \left( \sqrt{c^2} + \frac{1}{\sqrt{\tau_j}} \right) \nonumber\\
&=& c^p+\left(\frac{1}{\sqrt{\tau_1}}+\frac{1}{\sqrt{\tau_2}}+\cdots+\frac{1}{\sqrt{\tau_p}}\right)c^{p-1}+\left(\frac{1}
{\sqrt{\tau_1\tau_2}}+\cdots+\frac{1}{{\tau_i \tau_j}}+\cdots\right)c^{p-2}+\nonumber\\
& & \cdots+\frac{1}{\sqrt{\tau_1\tau_2...\tau_p}}. \label{eqtrace3}
\end{eqnarray}

\noindent
It follows from (\ref{eqtrace3}) that 
\begin{eqnarray}
& & \int_{\mathbb{R}_{+}^{p}} |{\bf A}_{\bm{\tau}}|^\frac{1}{2} \exp \left(-\frac{\lambda^2}{6}\sum_{j=1}^{p}\tau_j 
\right) \, d{\bm{\tau}} \nonumber\\
&\leq& c^p \int_ {\mathbb{R}_{+}^{p}} \exp \left( -\frac{\lambda^2}{6}\sum_{j=1}^{p}\tau_j \right) \,d{\bm{\tau}}+c^{p-1}\int_ 
{\mathbb{R}_{+}^{p}}\left(\frac{1}{\sqrt{\tau_1}}+\frac{1}{\sqrt{\tau_2}}+\cdots+\frac{1}{\sqrt{\tau_p}}\right) \exp 
\left(-\frac{\lambda^2}{6}\sum_{j=1}^{p} \tau_j \right)\,d{\bm{\tau}}+ \nonumber\\
& & c^{p-2} \int_{\mathbb{R}_{+}^{p}}\left(\frac{1}{\sqrt{\tau_1\tau_2}}+\cdots+\frac{1}{\sqrt{\tau_i\tau_j}}+\cdots\right) \exp \left(-
\frac{\lambda^2}{6}\sum_{j=1}^{p}\tau_j \right)\,d{\bm{\tau}}+\nonumber\\
& & \cdots+\int_{\mathbb{R}_{+}^{p}}\frac{1}{\sqrt{\tau_1\tau_2...\tau_p}} \exp \left(-\frac{\lambda^2}{6}\sum_{j=1}^{p}\tau_j \right) 
\, d{\bm{\tau}}. \label{eqtrace4}
\end{eqnarray}

\noindent
Now, note that for any $m \in \{1,2,\cdots,p\}$, 
\begin{eqnarray}
& & \int_{\mathbb{R}^p_+} \frac{1}{\sqrt{\tau_{i_1}} \sqrt{\tau_{i_2}} \cdots \sqrt{\tau_{i_m}}} \exp \left(-\frac{\lambda^2}{6}
\sum_{j=1}^{p}\tau_j \right)\,d{\bm{\tau}} \nonumber\\
&=& \left( \frac{6}{\lambda^2} \right)^{p-m} \left(\int_{\mathbb{R}_{+}} \frac{1}{\sqrt{\tau_{i_1}}} e^{-\frac{\lambda^2}{6} 
\tau_{i_1}}\,d{{\tau_{i_1}}}\right)\times\cdots\times \left(\int_{\mathbb{R}_{+}}\frac{1}{\sqrt{\tau_{i_m}}} e^{-\frac{\lambda^2}
{6}\tau_{i_m}}\,d{{\tau_{i_m}}}\right) < \infty. 
\label{eqtrace5}
\end{eqnarray}

\noindent
The result now follows by (\ref{eqtrace2}), (\ref{eqtrace4}) and (\ref{eqtrace5}). 
\end{proof}

\section{Additional Numerical Results}
\label{sec:more-numerical}

Figure~\ref{fig:rs} is similar to
the left-hand side of Figure~\ref{fig:dacf}
under settings of high multicollinearity (left), low sparsity (center), and both high multicollinearity and low sparsity (right).
(See the following section for details.)
It is clear that the behavior seen in Figure~\ref{fig:ac-p} is also seen in other settings of multicollinearity and sparsity.

\begin{figure}[htbp]
\centering
\includegraphics[scale=0.38]{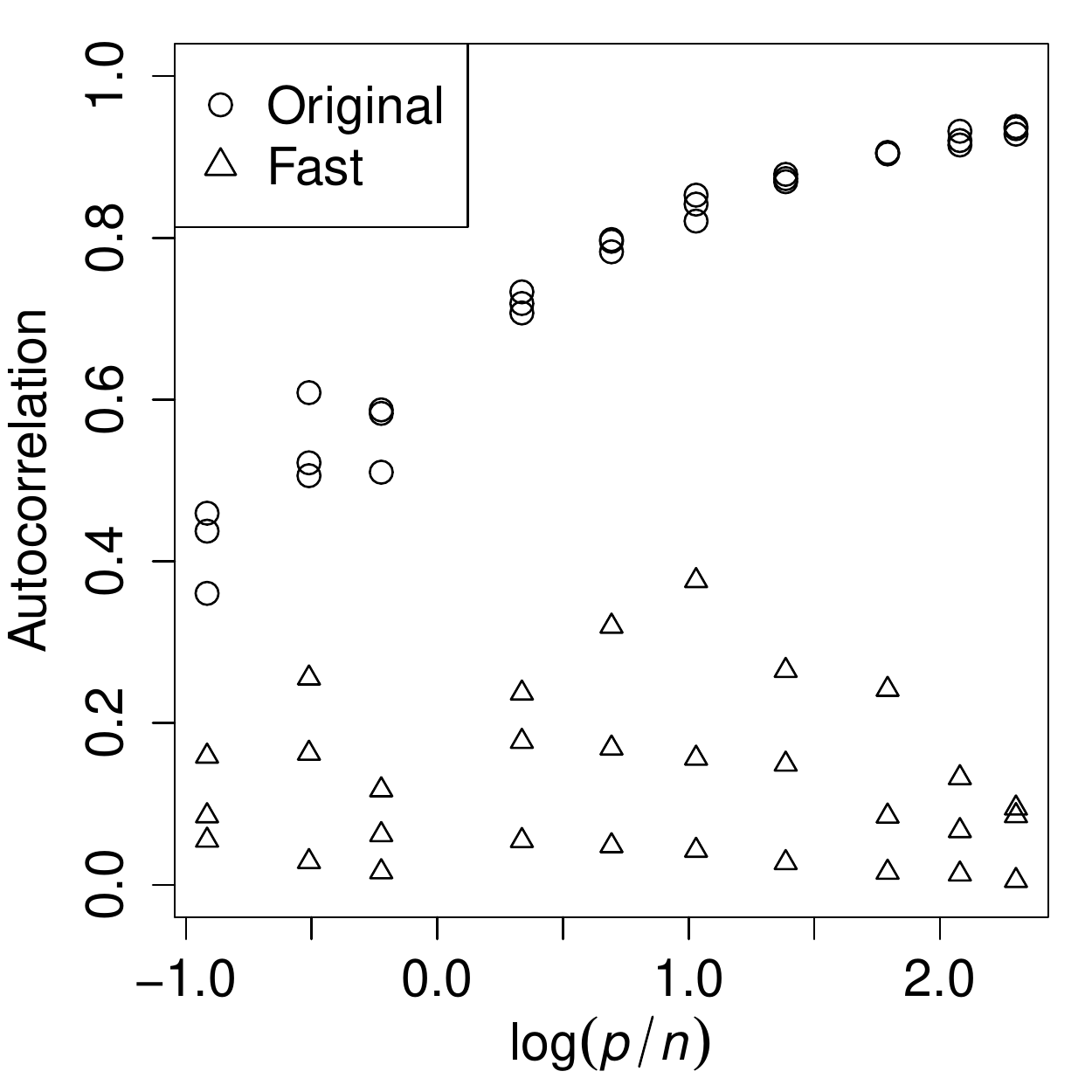}
\includegraphics[scale=0.38]{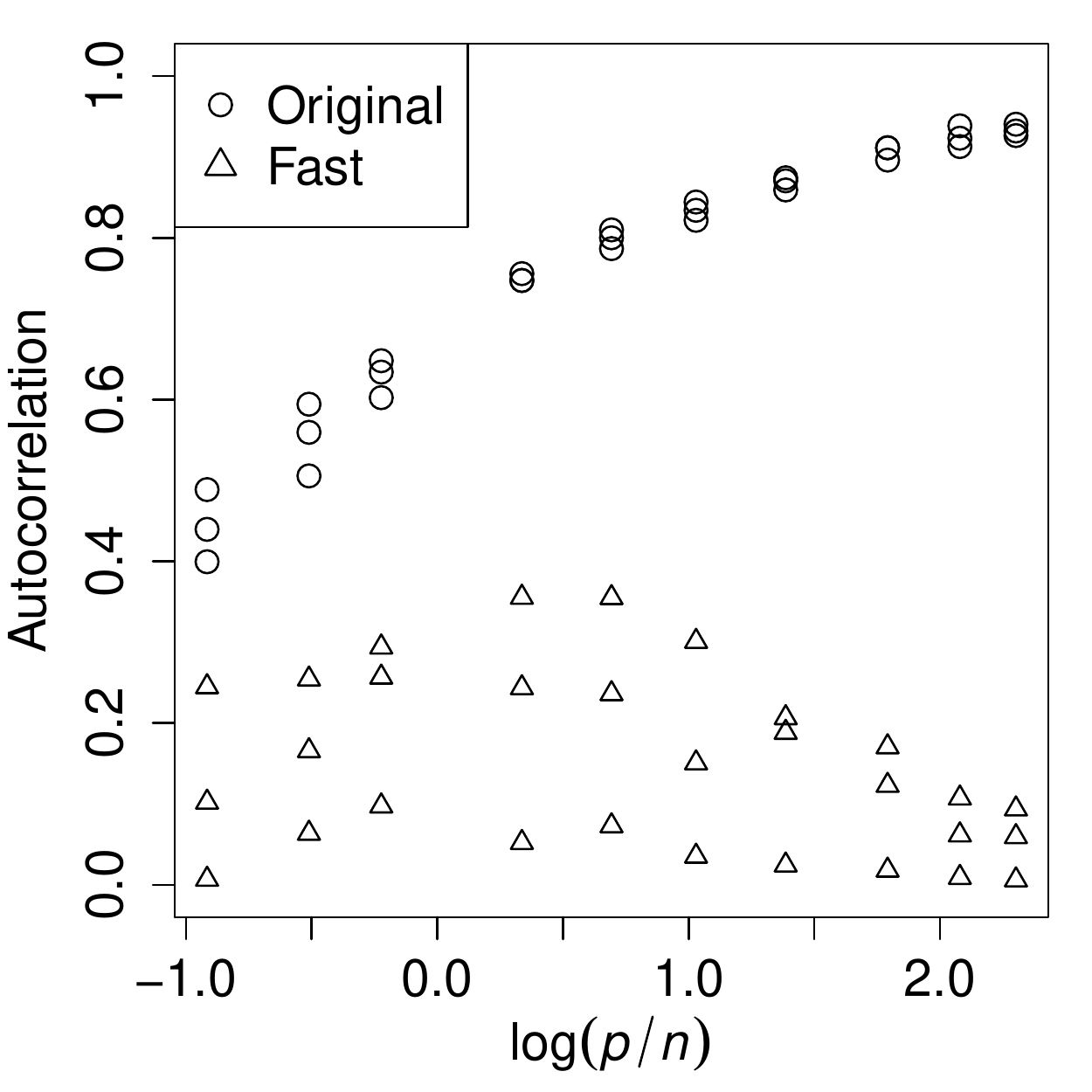}
\includegraphics[scale=0.38]{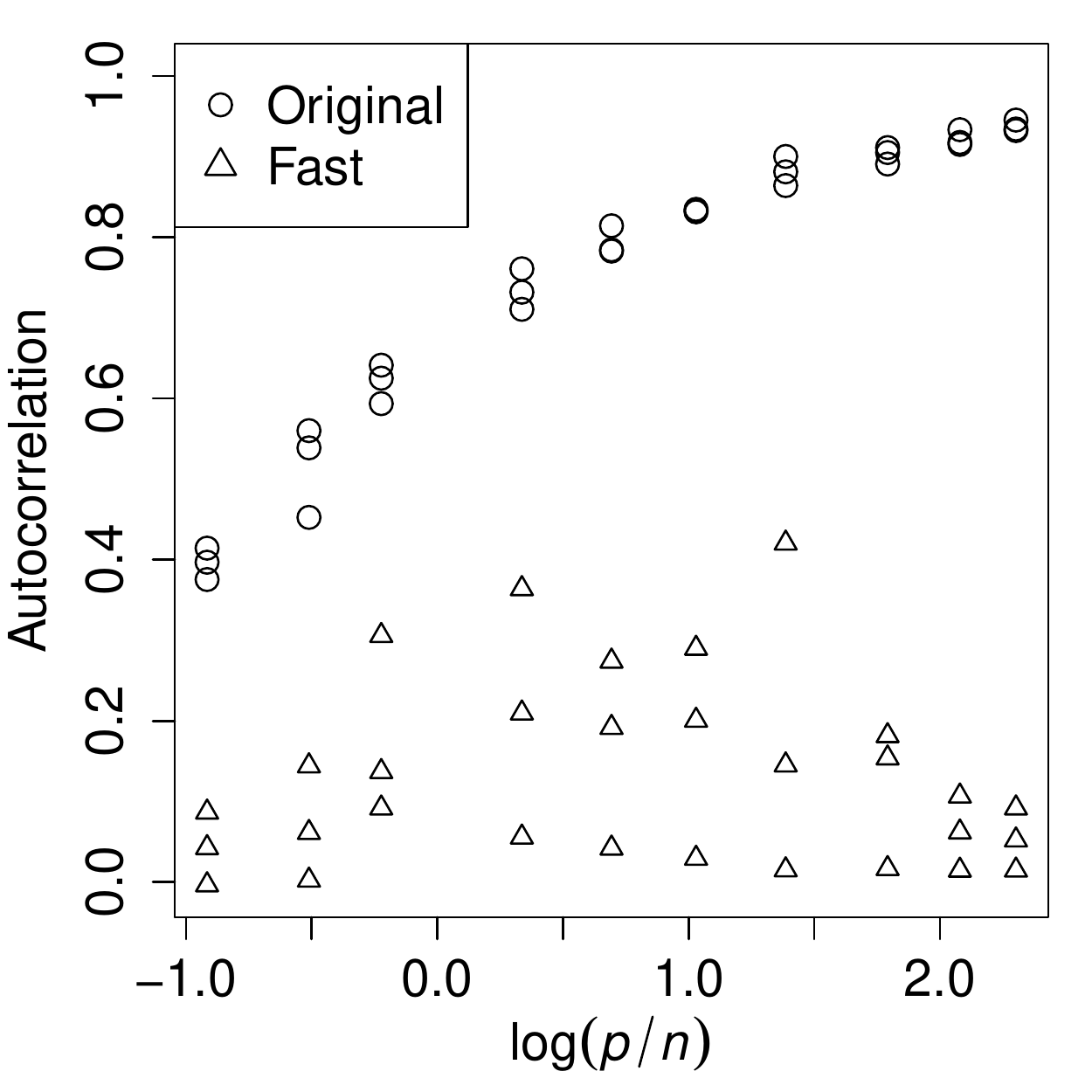}
\vspace*{-1em}
\caption{%
Autocorrelation of the $\sigma^2_k$ chain versus $p/n$ for various values of~$n$ for the original and two-step Bayesian lasso Gibbs samplers in settings of high multicollinearity (left), low sparsity (center), and both high multicollinearity and low sparsity (right).
See the
following section
for details of the generation of the
numerical quantities used
in the execution of these chains.
}
\label{fig:rs}
\end{figure}

\section{%
Details of Numerical Results
}
\label{sec:numerical-details}

Each plotted point
in Figures~\ref{fig:ac-p}~and~\ref{fig:dacf}
represents the average lag-one autocorrelation over
10 Gibbs sampling runs of 10,000 iterations each
(after discarding an initial ``burn-in'' period of 1,000 iterations).
For each of the 10~runs at each $n$ and $p$ setting, the $np$ elements of the $n\times p$ covariate matrix~$\bm X$ were
first drawn as $N(0,1)$ random variables with all pairwise correlations equal to $1/5$.
The columns of~$\bm X$ were then standardized to have mean zero and squared Euclidean norm~$n$.
Also, for each run, the $n\times1$ response vector~$\bm Y$ was generated as $\bm Y=\bm X\bm\beta_\star+\bm\epsilon$,
where $\bm\beta_\star$ is a $p\times1$ vector with its first
$\lceil p/5\rceil$
elements drawn
as independent $t_2$ random variables
and its remaining
$p-\lceil p/5\rceil$
elements set to zero, and where $\bm\epsilon$ is an $n\times1$ vector of independent $t_4$ random
variables.
The initial values were set as $\bm\beta_0=\bm1_p$ and $\sigma^2_0=1$.
The regularization parameter~$\lambda$ was set to $\lambda=1$.  
In Figure~\ref{fig:ac-p}, for each~$n$, the values of $p$ were $2n$, $4n$, $6n$, $8n$, and $10n$.  The left side of Figure~\ref{fig:dacf} includes the same results as in Figure~\ref{fig:ac-p} but also includes $p$ equal to $2n/5$, $3n/5$, $4n/5$, $7n/5$, and $14n/5$.  The DACF surface plots in the center and right side of Figure~\ref{fig:dacf} used each combination of $n,p\in\{5,15,25,35,45\}$.

For the examples in Section~\ref{sec:applications}, the starting points of all chains were set as $\bm\beta_0=\bm1_p$ and $\sigma^2_0=1$.
For the crime data in Subsection~\ref{subsec:crime},
the $p=1325$ covariates were the 50 first-order terms, 50 quadratic terms, and 1225 multiplicative interaction terms involving 50 first-order covariates, which themselves were chosen as follows.  First, all potential covariates with missing data were excluded, which left 99 potential covariates.  Next, any remaining potential covariates for which over half of the values were equal to any single value were excluded, which left 94 potential covariates.  From the remaining potential covariates, we selected the 50 with the largest absolute correlation with the response vector.
These 50 chosen covariates were then standardized to have mean zero before the second-order terms were computed.
From the 1994 total observations,
$n=10$
were chosen at random.  The final columns of the design matrix (now with
$n=10$
rows) were then standardized again to have mean zero and squared Euclidean norm $n$.

For the numerical results of the spike-and-slab sampler in Subsection~\ref{subsec:numerical-ss}, 
all settings were identical to those used for the corresponding results for the Bayesian lasso, with the obvious exception of the hyperparameters in the prior.
The hyperparameters of the spike-and-slab prior were set as $w_j=1/2$ and $\kappa_j=100$ for each $j\in\{1,\ldots,p\}$.  For the gene expression data, we set
$\zeta_j=0.00002$
for each $j\in\{1,\ldots,p\}$.
For the simulated data, we set $\zeta_j=0.01$ for each $j\in\{1,\ldots,p\}$.

The plots in Figure~\ref{fig:rs} were constructed similarly to the left-hand side of Figure~\ref{fig:dacf}, but
with
the following additional modifications.  For the left-hand side and the right-hand side (but not the center), the pairwise correlation between each element of the design matrix was~$4/5$ (rather than~$1/5$).
For the center and right-hand side (but not the left-hand side), the number of nonzero coefficients for the generation of the response vector was taken to be~$\lceil 4p/5\rceil$ (rather than~$\lceil p/5\rceil$).

The R function \texttt{set.seed} was used to initialize the random number generation.  Different seeds were used for each data set and each collection of simulations to facilitate reproducibility of individual figures and results.  The seeds themselves were taken as consecutive three-digit blocks of the decimal expansion of $\pi$ (i.e., 141, 592, 653, etc.) to make absolutely clear that they were not manipulated to gain any advantage for the proposed methodology.
The reader can of course verify directly that the results are qualitatively similar for other seeds.

\section{R Code}
\label{sec:code}

\nocite{R}
The R code for generating the numerical results of Sections~\ref{sec:numerical},~\ref{sec:applications},~and~\ref{sec:more-numerical} is provided in its entirety in an accompanying file.
This code incorporates functions or data from the R packages \texttt{coda} 
\citep{Rcoda}, \texttt{flare} \citep{Rflare}, \texttt{lars} \citep{Rlars}, 
\texttt{ppls} \citep{Rppls}, and \texttt{statmod} \citep{Rstatmod}.
We also reproduce below the portion of the R code that executes the original and two-step Bayesian lasso algorithms themselves.
\footnotesize
\begin{verbatim}
# Slightly modified version of rinvgauss function from statmod package
rinvgauss. <- function(n,mean=1,shape=NULL,dispersion=1){
	if(!is.null(shape)){dispersion <- 1/shape}
	mu <- rep_len(mean,n)
	phi <- rep_len(dispersion,n)
	r <- rep_len(0,n)
	i <- (mu>0 & phi>0)
	if(!all(i)){
		r[!i] <- NA
		n <- sum(i)
	}
	phi[i] <- phi[i]*mu[i]
	Y <- rchisq(n,df=1)
	X1 <- 1+phi[i]/2*(Y-sqrt(4*Y/phi[i]+Y^2))
	# Note: The line above should yield all X1>0, but it occasionally doesn't due to
	#		numerical precision issues.  The line below detects this and recomputes
	#		the relevant elements of X1 using a 2nd-order Taylor expansion of the 
	#		sqrt function, which is a good approximation whenever the problem occurs.
	if(any(X1<=0)){X1[X1<=0] <- (1/(Y*phi[i]))[X1<=0]}
	firstroot <- as.logical(rbinom(n,size=1L,prob=1/(1+X1)))
	r[i][firstroot] <- X1[firstroot]
	r[i][!firstroot] <- 1/X1[!firstroot]
	mu*r
}

# Draw from inverse-Gaussian distribution while avoiding potential numerical problems
rinvgaussian <- function(n,m,l){
	m. <- m/sqrt(m*l)
	l. <- l/sqrt(m*l)
	sqrt(m*l)*rinvgauss.(n,m.,l.)
}

# Gibbs iteration functions for both Bayesian lassos
# Note: The versions have separate functions, as opposed to being different
#       options of the same function, since the latter would require checking any such
#       options every time the function is called, i.e., in every MCMC iteration.
# Note: The values of XTY, n, and p can obviously be calculated from X and Y, but they
#       are included as inputs to avoid recalculating them every time the function is
#       called, i.e., in every MCMC iteration.
iter.bl.original <- function(beta,sigma2,X,Y,XTY,n,p,lambda){
	d.tau.inv <- rinvgaussian(p,sqrt(lambda^2*sigma2/beta^2),lambda^2)
	A.chol <- chol(t(X)%*%X+diag(d.tau.inv))
	beta.tilde <- backsolve(A.chol,backsolve(t(A.chol),XTY,upper.tri=F))
	Z <- rnorm(p)
	beta.new <- beta.tilde+sqrt(sigma2)*backsolve(A.chol,Z)
	sigma2.new <- (sum((Y-drop(X%*%beta.new))^2)+
				   sum(beta.new^2*d.tau.inv))/rchisq(1,n+p-1)
	return(list(beta=beta.new,sigma2=sigma2.new))
}
iter.bl.fast <- function(beta,sigma2,X,Y,XTY,n,p,lambda){
	d.tau.inv <- rinvgaussian(p,sqrt(lambda^2*sigma2/beta^2),lambda^2)
	A.chol <- chol(t(X)%*%X+diag(d.tau.inv))
	beta.tilde <- backsolve(A.chol,backsolve(t(A.chol),XTY,upper.tri=F))
	sigma2.new <- (sum(Y^2)-sum(XTY*beta.tilde))/rchisq(1,n-1)
	Z <- rnorm(p)
	beta.new <- beta.tilde+sqrt(sigma2.new)*backsolve(A.chol,Z)
	return(list(beta=beta.new,sigma2=sigma2.new))
} 

# Run original and two-step Bayesian lassos
run.bl <- function(X,Y,lambda,K,M,outfile.stem,fast=F,keep.beta=T,write.each=F){
	XTY <- drop(t(X)%*%Y)
	n <- dim(X)[1]
	p <- dim(X)[2]
	iter.bl <- get(paste("iter.bl.",ifelse(fast,"fast","original"),sep=""))
	chaindigits <- max(1,ceiling(log(M,10)))  # digits needed for chain label strings
	for(chain in 0:(M-1)){
		beta <- rep(1,p)
		sigma2 <- 1
		chaintext <- substring(format(chain/(10^chaindigits),nsmall=chaindigits),3)
		outfile.beta <- paste(outfile.stem,"-",chaintext,"-b.txt",sep="")
		outfile.sigma2 <- paste(outfile.stem,"-",chaintext,"-s.txt",sep="")
		if(write.each){
			for(k in 1:K){
				iter.result <- iter.bl(beta,sigma2,X,Y,XTY,n,p,lambda)
				beta <- iter.result$beta
				sigma2 <- iter.result$sigma2
				if(keep.beta){
					beta.row <- matrix(beta,nrow=1)
					write.table(beta.row,outfile.beta,append=T,row.names=F,col.names=F)
				}
				write.table(sigma2,outfile.sigma2,append=T,row.names=F,col.names=F)
			}
		}else{
			beta.chain <- matrix(NA,nrow=K,ncol=p)
			sigma2.chain <- rep(NA,K)
			for(k in 1:K){
				iter.result <- iter.bl(beta,sigma2,X,Y,XTY,n,p,lambda)
				beta <- iter.result$beta
				sigma2 <- iter.result$sigma2
				beta.chain[k,] <- beta
				sigma2.chain[k] <- sigma2
			}
			if(keep.beta){
				write.table(beta.chain,outfile.beta,row.names=F,col.names=F)
			}
			write.table(sigma2.chain,outfile.sigma2,row.names=F,col.names=F)
		}
		typetext <- ifelse(fast,"Fast","Orig")
		print(paste(typetext,"chain",chain+1,"of",M,"complete at",date()))
		flush.console()
	}
}
\end{verbatim}
\normalsize

\bibliographystyle{ims}
\bibliography{references}

\end{document}